\documentclass[journal]{IEEEtran}
\usepackage{cite}
\usepackage{amsmath,amssymb,amsfonts, yhmath}
\usepackage{amsthm}
\usepackage{graphicx}
\usepackage{subfigure}   
\usepackage{booktabs}
\usepackage{textcomp}
\usepackage{xcolor}
\usepackage{algorithm}
\usepackage[noend]{algorithmic}    
\usepackage{framed,multirow}
\usepackage{mwe}
\usepackage{diagbox}
\usepackage{bbding}
\usepackage{fancyhdr}
\usepackage{bbm}
\usepackage{bm}

\makeatletter
\newcommand{\removelatexerror}{\let\@latex@error\@gobble}
\makeatother

\theoremstyle{plain}
\newtheorem{thm}{Theorem}

\newtheorem{lemma}{Lemma}
\newtheorem{definition}{Definition}

\theoremstyle{definition}
\newtheorem{remark}{Remark}

\usepackage[switch,columnwise]{lineno}

\usepackage{hyperref}  

\IEEEoverridecommandlockouts
\begin{document}

\title{Achieving $\alpha$-Fairness in Clustered Cell-Free Networking: A Tight Relaxation Approach
}

\author{\IEEEauthorblockN{Chaowen Deng, Jie Fan, Boxiang Ren, Ziyuan Lyu, Jingchen Peng,\\ 
Hao Wu and Junyuan Wang,~\IEEEmembership{Member,~IEEE}}
\thanks{This paper was presented in part in IEEE Wireless Communications and Networking Conference (WCNC), Milan, Italy, March 2025~\cite{10978229}.  (Corresponding author: Junyuan Wang.)}
\thanks{Chaowen Deng, Jie Fan, Boxiang Ren, Ziyuan Lyu, Jingchen Peng and Hao Wu are with Department of Mathematical Sciences, Tsinghua University, Beijing, China (email: dcw21@mails.tsinghua.edu.cn, \{fanj21, rbx21, lvzy21\}@tsinghua.org.cn), pjc22@mails.tsinghua.edu.cn, hwu@tsinghua.edu.cn.}
\thanks{J. Wang is with the College of Electronic and Information Engineering, Tongji University, Shanghai, China (email: junyuanwang@tongji.edu.cn).}}

\maketitle

\begin{abstract}

Clustered cell-free networking has emerged as a promising architecture to balance the high performance of cell-free massive MIMO and the scalability of traditional cellular systems. However, achieving fairness across subnetworks remains a critical yet largely unsolved challenge. This paper investigates the fairness problem in clustered cell-free networking and proposes a unified and tunable \(\alpha\)-fairness scheme that effectively balances overall spectral efficiency and inter-subnetwork fairness. Using the closed-form deterministic equivalent of the ergodic sum capacity, we reformulate the combinatorial clustering problem as a continuous optimization problem. Leveraging the concavity/convexity properties of the \(\alpha\)-fair objective, we classify the problem into four distinct cases according to the value of \(\alpha\). For each case, we establish the exact equivalence between the original integer program and its continuous relaxation, and develop efficient algorithms with guaranteed convergence. Extensive simulations show that the proposed scheme achieves up to 11\% improvement in Jain’s fairness index and 45\% gain in minimum subnetwork capacity, with only a negligible 5\% reduction in aggregate throughput. 

\end{abstract}


\section{Introduction}

\IEEEPARstart{T}{he} evolution from the fifth-generation (5G) to the sixth-generation (6G) of mobile communication systems is driving a profound integration of the physical and digital worlds, precipitating an explosive growth in the demand for wireless connectivity \cite{8869705}, not only for ultra-high data transmission rates but also for consistent, high-reliability, and low-latency services across the entire coverage area. However, this evolution is currently hitting a physical bottleneck inherent to traditional cellular architectures, namely the cell-edge problem \cite{5770660}.

In conventional cellular networks \cite{2}, each base station (BS) serves users within a fixed cell. As BS density increases, cell sizes shrink, causing more users to experience the cell-edge problem due to weak serving signals and strong inter-cell interference (ICI) \cite{340450}. Although coordinated multi-point (CoMP) transmission \cite{5706317} can partially mitigate ICI, it fails to fully resolve the issue since some users inevitably lie at cluster edges. Cell-free massive MIMO (CF-mMIMO) \cite{7827017, 8845768, interdonato2019ubiquitous} overcomes this limitation by coordinating a large number of distributed access points (APs) through a central processing unit (CPU), thereby eliminating cells and the cell-edge problem. However, CF-mMIMO suffers from a severe scalability bottleneck \cite{9064545}, as the overhead of global channel state information (CSI) acquisition, joint precoding, and fronthaul data exchange grows prohibitively with network size \cite{11220251}.

To strike a balance between the low complexity of cellular architectures and the high performance of full cooperation, clustered cell-free networking \cite{8007415, 10014663} has been proposed. The core idea of this architecture is to strategically decompose a massive network into multiple non-overlapping, parallel subnetworks. Within each subnetwork, BSs and users execute joint processing to eliminate internal interference. Meanwhile, inter-subnetwork interference is minimized through intelligent topology optimization, which confines strong interference links within the subnetworks. This architecture not only preserves the macro-diversity gains of cell-free systems but also ensures network scalability by constraining computational and signaling overheads within acceptable limits.

\IEEEpubidadjcol 

There have been a number of works investigating the clustered cell-free networking problem\cite{8007415, 10014663, 9838756, ren2024sequential, zeng2024tunable, tight, 10766356,9754261,9839089}, i.e., how to decompose a network into subnetworks. 
To guide the optimization of clustered cell-free networking, the strength of the total inter-subnetwork interference and the achievable rate were chosen in \cite{8007415} and  \cite{10014663,9838756,ren2024sequential,zeng2024tunable}, respectively, as the performance metric. In \cite{8007415}, the bipartite graph is employed to model wireless networks subject to inter-subnetwork interference, serving as the foundation for the graph-theoretic binary search based spectral relaxation (BSSR) algorithm. The authors of \cite{10014663,9838756} further provide users with more specific performance indicators by posing a minimum requirement on the average rate of all users, which provides QoS guarantee at the network level. To establish an equivalence between clustered cell-free networking and the min $k$-cut problem, \cite{ren2024sequential} introduces a sequential min $k$-cut method. By employing two strategic problem transformations, this approach effectively bridges the gap between wireless networking and graph theory. For beam-level clustered cell-free networking, \cite{zeng2024tunable} presents a tunable weighted kernel $k$-means algorithm, thereby accelerating the networking process while achieving flexible beam on/off control and improving spectral efficiency. Nonetheless, both interference strength and achievable rate vary with small-scale fading, leading to fast-varying network decomposition result and hence frequent handover, which is undesirable in practice. As such, the above works \cite{8007415,10014663,9838756,ren2024sequential,zeng2024tunable} just simply ignored small-scale fading in the channel model, which, however, results in degradation in rate performance.

Other studies \cite{tight, 10766356,9754261,9839089} considered ergodic capacity as the performance metric. In \cite{tight}, given the number of subnetworks, the sum ergodic capacity is maximized and a Bregman proximal gradient (BPG) based clustered cell-free networking algorithm is proposed. In \cite{10766356}, the number of subnetworks and the corresponding network decomposition are jointly optimized, aiming to maximize the sum ergodic capacity under a joint processing constraint. However, both of them concentrated on enhancing the sum ergodic capacity of the whole network, overlooking the capacity of each subnetwork. Consequently, a subnetwork may have much higher sum capacity than others, that is, the traffic loads could be significantly unbalanced across different subnetworks. In \cite{9754261}, the sum ergodic capacity of each subnetwork is required to exceed the predefined threshold and a capacity-centric (C$^{2}$) algorithm is proposed to maximize the number of subnetworks under this constraint. Although this provides performance guarantees at the subnetwork level, it overlooks fairness across subnetworks and may thus lead to significant load imbalance. The authors of \cite{9839089} investigated the fairness problem in clustered cell-free networking and proposed to maximize the minimum ergodic capacity of the subnetworks, but the proposed heuristic capacity-guaranteed networking (CGN) algorithm lacks rigorous theoretical performance guarantees.

In this paper, we further investigate the fairness properties of clustered cell-free networking. Specifically, we construct a unified and tunable model based on the $\alpha$-fairness utility function \cite{879343, jain1984quantitative} that effectively balances overall spectral efficiency and fairness across subnetworks, together with high-performance algorithms equipped with rigorous theoretical guarantees. By approximating the ergodic sum capacity of each subnetwork with its closed-form deterministic equivalent \cite{Hachem_2007}, we transform the original combinatorial clustering problem into a continuous optimization problem. Leveraging the concavity/convexity properties of the objective function formulated with the $\alpha$-fairness utility function, we rigorously classify the problem into four distinct cases according to the value of $\alpha$. For each case, we establish the equivalence between the original integer program and its continuous relaxation via tailored theoretical proofs, and develop efficient algorithms with guaranteed convergence. This systematic $\alpha$-fairness approach not only significantly outperforms existing clustering schemes in terms of Jain’s fairness index \cite{jain1984quantitative} and minimum subnetwork capacity, but also provides network operators with a single tunable parameter $\alpha$ to smoothly navigate the entire efficiency-fairness trade-off spectrum in practical fronthaul-constrained clustered cell-free massive MIMO deployments.

The major contributions of this paper are summarized as follows:

\begin{itemize}
    \item We formulate a novel optimization approach that unifies fairness in clustered cell-free networks. By leveraging a tunable $\alpha$-fair utility function, we bridge the gap between sum-capacity maximization ($\alpha=0$) and max-min fairness ($\alpha\to\infty$ \cite{879343}). This provides network operators with a flexible mechanism to navigate the trade-off between spectral efficiency and user fairness, addressing the limitations of rigid, single-objective clustering schemes.

    \item We derive a suite of convergence-guaranteed algorithms grounded in the exact relaxation of the combinatorial problem. By utilizing deterministic equivalents, the original combinatorial clustering problem is transformed into a continuous program. We rigorously establish the equivalence between the discrete and continuous formulations and propose efficient algorithms with guaranteed convergence.

    \item Extensive numerical results validate the efficacy of our approach. In comparison to state-of-the-art baselines, our method yields substantial gains in fairness and edge performance, improving Jain's fairness index by up to 11\% and minimum subnetwork capacity by 45\%, while incurring a negligible 5\% loss in aggregate throughput. This confirms the proposed solution as a robust candidate for practical, fronthaul-constrained cell-free deployments.
\end{itemize}

The following content of the paper is organized as follows. The remainder of this paper is organized as follows. The uplink $\alpha$-fairness based clustered cell-free system model is described in Section \ref{sec:Network Model}. Section \ref{sec:Alpha Fairness} establishes the exact equivalence between the original integer problem and its continuous relaxation through a case-by-case analysis, and develops the corresponding optimization algorithms. Section \ref{sec:Simulation Results} presents extensive simulation results, providing a performance comparison across varying $\alpha$ configurations. Finally, Section \ref{sec:Conclusion} concludes the paper, while detailed mathematical proofs are relegated to the Appendices.

$Notation$: In this paper, scalars are denoted by lowercase letters (e.g., \(h\)), vectors by bold lowercase letters (e.g., \(\mathbf{h}\)), and matrices by bold uppercase letters (e.g., \(\mathbf{H}\)). Calligraphic letters represent sets (e.g., \(\mathcal{S}_m\)). The notation \(\mathcal{CN}(0,1)\) refers to the complex Gaussian distribution with zero mean and unit variance. The operators \((\cdot)^T\), \((\cdot)^\dagger\), \(\det(\cdot)\), \(\operatorname{tr}(\cdot)\), \(\nabla\), \(\mathbb{E}\{\cdot\}\), \(\langle \cdot, \cdot \rangle\), \(\frac{\cdot}{\cdot}\), and \(\circ\) denote the transpose, Hermitian transpose, determinant, trace, gradient, expectation, inner product, entrywise division, and Hadamard product, respectively. The symbols \(\mathbb{R}_+\), \(\mathbb{C}\), and \(\mathbb{N}^{*}\) represent the set of non-negative real numbers, the set of complex numbers, and the set of positive integers, respectively.

\section{Network Model And Problem Formulation}\label{sec:Network Model}

\subsection{System Model}
Consider a wireless network consisting of $L$ single-antenna BSs and $K$ single-antenna users. The set of BSs is denoted by $\mathcal{B} = \{b_{1},\cdots, b_{L}\}$ and the set of users is $\mathcal{U} = \{u_{1},\cdots, u_{K}\}$ with $|\mathcal{B}| = L$ and $|\mathcal{U}| = K$. Assume that the network is decomposed into $M$ non-overlapping subnetworks and the $m$-th ($m=1,2,\cdots,M$) subnetwork is the set of BSs and users, denoted by $\mathcal{S}_m$ with $\cup_{m=1}^{M}\mathcal{S}_{m} = \mathcal{B} \cup \mathcal{U}$ and $\mathcal{S}_{m} \cap \mathcal{S}_{m'} = \emptyset$ for any $m \neq m'$. Network decomposition $\mathcal{M} = \{\mathcal{S}_{1}, \mathcal{S}_{2}, \cdots, \mathcal{S}_{M}\}$ forms a decomposition of set $\mathcal{B} \cup \mathcal{U}$. In $\mathcal{S}_{m}$, the number of users and the number of BSs are denoted by $K_{m}$ and $L_{m}$, respectively, with $K = \sum_{m=1}^{M}K_{m}$ and $L = \sum_{m=1}^{M}L_{m}$.

Let us denote the channel gain coefficient from user $u_{k} \in \mathcal{U}$ to BS $b_{l} \in \mathcal{B}$ by $h_{lk}$, and model it as
\begin{equation}
    \label{h}
    h_{lk} = q_{lk}g_{lk},
\end{equation}
    where $g_{lk}$ is the small-scale fading coefficient
    and $q_{lk}$ represents the large-scale fading coefficient. By assuming that each user transmits with power $P$, the uplink sum ergodic capacity of subnetwork $\mathcal{S}_{m}$ can be expressed as\cite{9754261}
\begin{equation}
    C_{m} \!=\! \mathbb{E} \left\{ \log \det \!\left[ \mathbf{I}+P \left( N_{0}\mathbf{I} + P\mathbf{\Pi}_{m}\mathbf{\Pi}_{m}^{\dagger} \right) ^{-1}\mathbf{H}_{m}\mathbf{H}_{m}^{\dagger} \right] \right\},\label{sum_capacity}  
\end{equation}
    where $\mathbf{H}_{m} \in \mathbb{C}^{L_{m} \times K_{m}}$ collects the channel gain coefficients between the users and the BSs in subnetwork $\mathcal{S}_{m}$ and $\boldsymbol{\Pi}_{m}\in \mathbb{C}^{L_{m} \times (K-K_{m})}$ collects the the channel gain coefficients from the users outside subnetwork $\mathcal{S}_{m}$ to the BSs in subnetwork~$\mathcal{S}_{m}$.  $N_{0}$ is the power of the additive white Gaussian noise (AWGN).

\subsection{Problem Formulation}\label{PF}
Fairness has recently attracted considerable attention in the design of wireless network systems \cite{7986959, 10706102, 10505156, 6517050}, including max-min fairness \cite{7069272, 7506136, 7460209}, proportional fairness \cite{10714036}, and other variants. To unify these distinct fairness metrics within a single flexible method, we construct a complete and tunable model based on $\alpha$-fairness \cite{879343, jain1984quantitative}. Specifically, we strategically decompose the entire network into $M$ subnetworks to ensure that the system satisfies prescribed fairness criteria with respect to the sum capacity of every subnetwork. The problem can be formulated as
\begin{subequations}
    \begin{align}
        \mathcal{P}1:
        \max_{\mathcal{M}} & \quad\sum_{m=1}^{M}u_{\alpha}(C_{m})&& \label{problemP1a}\\
        \mbox{s.t.}
        &\quad |\mathcal{S}_{m} \cap \mathcal{U} | \geq 1 ,&&\forall m,\label{problemP1b} \\
        &\quad |\mathcal{S}_{m} \cap \mathcal{B} | \geq 1, &&\forall m,\label{problemP1c} \\
        &\quad \mathcal{S}_{m} \cap \mathcal{S}_{m'} = \emptyset, &&\forall m' \neq m,\label{problemP1d} \\
       &\quad \mathop{\bigcup}\limits_{m=1}^{M}\mathcal{S}_{m} = \mathcal{U} \cup \mathcal{B}.  \label{problemP1e}
    \end{align}
\end{subequations}
The constraints \eqref{problemP1b} and \eqref{problemP1c} guarantee that each subnetwork has at least one user and one BS, and constraints \eqref{problemP1d} and \eqref{problemP1e} ensure that the subnetworks are non-overlapping. The function $u_{\alpha}(C_{m})$ is $\alpha$-fairness utility \cite{879343, jain1984quantitative}, which is determined by
\begin{equation}
    \label{alpha_fairness}
    u_{\alpha}(C_{m}) \triangleq
    \begin{cases}
        \ln (C_{m}), & \alpha = 1,\\
        \dfrac{C_{m}^{1-\alpha}}{1-\alpha}, & \alpha \neq 1, \alpha \geq 0.
    \end{cases}
\end{equation}
\IEEEpubidadjcol 
By adjusting the parameter $\alpha$, the problem $\mathcal{P}1$ enables a smooth transition between ``pure efficiency" ($\alpha$ = 0) and ``absolute fairness" ($\alpha \rightarrow \infty$). In general, increasing $\alpha$ results in higher fairness \cite{5461911}. For instance, maximum sum ergodic capacity of the network can be achieved by setting $\alpha = 0$ with no fairness, which was studied in \cite{tight}, whereas proportional and max-min fairness can be achieved by setting $\alpha =1$ \cite{kelly1998rate} and $\alpha \rightarrow \infty$ \cite{879343}.

Note that decomposing a network requires to cluster both BSs and users. As the BS partition is inherently stable in practical systems due to complexity and physical limitations \cite{demir2021foundations}, in this paper, we will direct our dedication towards the clustering of users for a given BS partition. With the aforementioned considerations, our dedication is directed towards the detailed clustering of users within the confines of a predefined BS partition. Sum capacity $C_{m}$ in \eqref{sum_capacity} can be rewritten as
\begin{equation}\label{sum_capacity_split}
\begin{split}
        C_{m}=\mathbb{E}& \left[ \log \det \left( \mathbf{I}+\dfrac{P}{N_{0}} \widetilde{\mathbf{H}}_{m}\widetilde{\mathbf{H}}_{m}^{\dagger} \right) \right]\\
        &\qquad \qquad - \mathbb{E}\left[ \log \det \left( \mathbf{I}+\dfrac{P}{N_{0}} \boldsymbol{\Pi}_{m}\boldsymbol{\Pi}_{m}^{\dagger}\right) \right],
\end{split}    
\end{equation}
where $\widetilde{\mathbf{H}}_{m} = [\mathbf{H}_{m}, \boldsymbol{\Pi}_{m}] \in \mathbb{C}^{L_{m}\times K}$. In that case, the first term on the right hand side of \eqref{sum_capacity_split} is a constant with respect to a given BS partition, denoted by 
\begin{equation}
     \mathbb{E}\left\{ \log \det \left[ \mathbf{I}\; +\; (P/N_{0})\widetilde{\mathbf{H}}_{m}\widetilde{\mathbf{H}}_{m}^{\dagger} \right] \right\} \triangleq c_{m}.
\end{equation}

\subsection{Problem Transformation}
Although the first term of \eqref{sum_capacity_split} is a constant, the second term involves the expectation of the logarithm of a determinant over random matrices posing significant challenges. To address this, we approximate $C_{m}$ by its deterministic equivalent, which leads to the reformulation of problem $\mathcal{P}1$. Specifically, let us first introduce a 0-1 matrix $\mathbf{X} = [x_{km}] \in \{0, 1\}^{K\times M}$ defined by
\begin{equation}
    \label{indicator_matrix_X}
    x_{km} = 
    \begin{cases}
        1, &\text{if~} u_{k} \in \mathcal{S}_{m},\\
        0, &\text{if~} u_{k} \notin \mathcal{S}_{m},
    \end{cases}
\end{equation}
and then rewrite $\mathcal{P}1$ as
\begin{subequations}
    \begin{align}
    \mathcal{P}2:
    \min_{\mathbf{X}} &\quad F_{\alpha}(\mathbf{X}) \triangleq -\sum_{m=1}^{M}u_{\alpha}(-f_{m}(\mathbf{X}))&& \label{problemP2a}\\
    \mbox{s.t.} 
    &\quad \sum_{k=1}^{K} x_{km} \geq 1, &&\forall m,\label{problemP2b} \\
    &\quad \sum_{m=1}^{M} x_{km} = 1, &&\forall k,\label{problemP2c} \\
    &\quad \mathbf{X} = [x_{km}] \in \{0, 1\}^{K\times M}. \label{problemP2d}
    \end{align}
\end{subequations}
For the constraints, \eqref{problemP2b} is equivalent to \eqref{problemP1b} in $\mathcal{P}1$, \eqref{problemP2c} is equivalent to \eqref{problemP1d} and \eqref{problemP1e} indicating that each user belongs to only one subnetwork, and \eqref{problemP1c} in $\mathcal{P}1$ naturally stands for a given feasible BS partition.
By adopting the deterministic equivalents based approximation of ergodic capacity \cite{Hachem_2007}, $f_{m}(\mathbf{X})$ can be approximated by
\begin{equation}
\begin{split}
    \label{sum_capacity_deterministic_equivalents}
    &f_{m}(\mathbf{X}) =  \mathbb{E}\left[ \log \det \left( \mathbf{I}+\dfrac{P}{N_{0}} \boldsymbol{\Pi}_{m}\boldsymbol{\Pi}_{m}^{\dagger}\right) \right] - c_{m} \\
    \approx &-  \sum_{b_{l}\in \mathcal{B}_{m}}\log \alpha_{l}^{m} - \sum_{k=1}^{K}\!\log \beta_{k}^{m} 
    - \sum_{b_{l}\in\mathcal{B}_{m}}\!\sum_{k=1}^{K}\theta^{m}_{lk}\alpha_{l}^{m}\beta_{k}^{m} - c_{m},
\end{split}    
\end{equation}
where $\boldsymbol{\alpha}^{m} \in \mathbb{R}^{L_{m}}$ and $\boldsymbol{\beta}^{m} \in \mathbb{R}^{K}$ are determined by
\begin{equation}
    \label{alpha_beta}
    \begin{aligned}
        &\alpha_{l}^{m} = \dfrac{1}{1+\sum_{k=1}^{K}\theta_{lk}^{m}\beta_{k}^{m}}, \ b_{l}\in\mathcal{B}_{m},\\
        &\beta_{k}^{m} = \dfrac{1}{1+\sum_{b_{l}\in \mathcal{B}_{m}}\theta_{lk}^{m}\alpha_{l}^{m}},  \ k=1,2, ..., K,
    \end{aligned}
\end{equation}
with $\theta_{lk}^{m} = (1-x_{km})\theta_{lk}$ and $\theta_{lk}= (Pq^{2}_{lk})/N_{0}$.

Considering the difficulties arising from the discrete nature of $\mathcal{P}2$, we will leverage a relaxation approach to transform it into a continuous problem, and this relaxation can theoretically be proven equivalent, with detailed discussion to follow in the next section. Specifically, relaxing the discrete $\mathbf{X}$ in constraint \eqref{problemP2d} into a continuous one, $\mathcal{P}2$ can be relaxed into the following continuous problem $\mathcal{P}3$ 
\begin{subequations}
    \begin{align}
        \mathcal{P}3:
        \min_{\mathbf{X}} & \quad F_{\alpha}(\mathbf{X})&& \label{problemP3a}\\
        \mbox{s.t.} 
        & \quad \mathbf{X} \in \varOmega, \label{problemP3b} 
    \end{align}
\end{subequations}
where $\varOmega = \varOmega_{1} \cap \varOmega_{2} \cap \varOmega_{3}$, with $\varOmega_{1} = \{\mathbf{X} | \mathbf{X} \geq 0\}$, $\varOmega_{2} = \{\mathbf{X} |\sum_{k=1}^{K}x_{km} \geq 1,\forall m\}$ and $\varOmega_{3} = \{\mathbf{X} |\sum_{m=1}^{M}x_{km}=1,\forall k\}$. Note that taking into account constraints $\mathbf{X}\in \varOmega_{2}$ and $\mathbf{X} \in \varOmega_{3}$, constraint $\mathbf{X} \in \varOmega_{1}$ is equivalent to $\mathbf{X} \in [0, 1]^{K\times M}$, i.e., each element of $\mathbf{X}$ is relaxed to a continuous variable in the range $[0, 1]$.

After relaxation to a continuous form, $f_{m}(\mathbf{X})$ in \eqref{sum_capacity_deterministic_equivalents} exhibits favorable analytical properties, including concavity and differentiability, which is crucial for efficient algorithm design.
\begin{lemma}[Concavity, Theorem 1 \cite{tight}]\label{fm_concave}
    $f_{m}(\mathbf{X}), \forall m$, defined in \eqref{sum_capacity_deterministic_equivalents} is concave in $\mathbf{X}$.
\end{lemma}

\begin{lemma}[Differentiability, Theorem 3 \cite{tight}]\label{partial_derivative}
    $f_{m}(\mathbf{X}), \forall m$, is differentiable with respect to $x_{km'},\forall k, m'$, and the partial derivative, i.e., the elements of $\nabla f_{m}(\mathbf{X})$ are given by
    \begin{equation}
    \label{fm_diff}
        \begin{split}
            \dfrac{\partial f_{m}}{\partial x_{km'}} = \left\{
            \begin{array}{ll}
                -\sum_{b_{l}\in \mathcal{B}_{m}}\beta_{k}^{m}\theta_{lk}\alpha_{l}^{m}, & m' = m,\\
                0, & m' \neq m,
            \end{array}
            \right.
        \end{split}
    \end{equation}
    where $\boldsymbol{\alpha}^{m}$ and $\boldsymbol{\beta}^{m}$ are determined by \eqref{alpha_beta}.
\end{lemma}

The gradient $\nabla F_{\alpha}(\mathbf{X})$ is then obtained as
\begin{equation}
    \label{Fderivative}
        \nabla F_{\alpha}(\mathbf{X}) = \sum_{m=1}^{M}\dfrac{1}{(-f_{m}(\mathbf{X}))^{\alpha}}\nabla f_{m}(\mathbf{X}).
\end{equation}

\section{Tight Relaxation and Algorithm Design for $\alpha$-Fair Clustered Cell-Free Networking}\label{sec:Alpha Fairness}
This section presents the algorithmic design for solving \(\mathcal{P}3\) and provides theoretical guarantees for the equivalence of its relaxation. Since the convexity and concavity properties of the objective function \(F_{\alpha}(\mathbf{X})\) change fundamentally with the value of \(\alpha\), which in turn directly affects the equivalence of relaxation, we rigorously classify \(\mathcal{P}3\) into four distinct cases over \(\alpha \in [0, \infty)\). For each case, we develop a specialized algorithm that exploits the specific structural properties of the objective to simultaneously guarantee theoretical equivalence between the original integer program and its continuous relaxation, as well as computational efficiency.


Notably, according to Theorem~2 in~\cite{tight}, when the objective function $F_{\alpha}(\mathbf{X})$ in problem $\mathcal{P}3$ is concave, problems $\mathcal{P}2$ and $\mathcal{P}3$ become equivalent in the sense that they share identical optimal objective values and admit a common optimal solution. However, the concavity of $F_{\alpha}(\mathbf{X})$ depends on the specific value of the fairness parameter $\alpha$. This observation motivates us to transform $F_{\alpha}(\mathbf{X})$ into a concave form through appropriate mathematical reformulations, thereby facilitating efficient solution of the problem. Specifically, we employ the convex-concave relaxation procedure (CCRP) \cite{liu2014graph} and the graduated nonconvexity and concavity procedure (GNCCP) \cite{liu2013gnccp} to achieve this transformation. Consequently, we categorize \(\mathcal{P}3\) into four distinct cases by partitioning the range of the fairness parameter \(\alpha\) into four disjoint sets $\mathcal{A}_1 \cup \mathcal{A}_2 \cup \mathcal{A}_3 \cup \mathcal{A}_4 = \mathbb{R}_{+} \cup \{\infty\}$, which is defined according to the convexity/concavity properties of the objective function \(F_{\alpha}(\mathbf{X})\) as follows: 
\begin{itemize}
    \item[1)] \bm{$Case$} \textbf{1:} \bm{$F_{\alpha}(\mathbf{X})$} \textbf{is a concave function} \bm{$(\alpha \in \mathcal{A}_1)$}.  In this case, problem \(\mathcal{P}3\) reduces to a concave minimization over the convex feasible set \(\varOmega\). For such structured problems, the classical Frank-Wolfe (FW) algorithm \cite{frank1956algorithm} is a standard and provably convergent method. Building upon this foundation, we propose the fast Frank-Wolfe (FFW) algorithm, which exploits an efficient solution method for the optimal production transport (OPT) problem~\cite{fan2024optimal} to significantly accelerate the projection subproblems while preserving the same convergence guarantees.
    
    \item[2)] \bm{$Case$} \textbf{2:} \bm{$F_{\alpha}(\mathbf{X})$} \textbf{is a convex function} \bm{$(\alpha \in \mathcal{A}_2)$}.  Direct optimization of the relaxed convex programming problem $\mathcal{P}3$ inherently fails to preserve equivalence guarantees with the original discrete formulation $\mathcal{P}2$. To bridge this theoretical gap, the CCRP will be employed to equivalently transform the convex programming problem $\mathcal{P}3$ into an optimization problem with a concave objective function, and the FFW algorithm is then utilized to solve the equivalently relaxed problem.

    \item[3)] \bm{$Case$} \textbf{3:} \bm{$F_{\alpha}(\mathbf{X})$} \textbf{is a non-convex and non-concave function} \bm{$(\alpha \in \mathcal{A}_3)$}.  In this case, a more sophisticated GNCCP is adopted to transform the objective function of $\mathcal{P}3$ from its original non-convex and non-concave form into a convex function first, then converts it to a concave formulation. The FFW algorithm is employed to derive the solution of the problem.
    
    \item[4)] \bm{$Case$} \textbf{4:} \bm{$\alpha \rightarrow \infty$} \bm{$(\alpha \in \mathcal{A}_4)$}.  This is  the special case called max-min fairness \cite{879343}. It can be transformed into a problem of finding the saddle point, which can be efficiently located using alternating gradient projection (AGP) algorithm \cite{pan2021efficient}.
\end{itemize}
\begin{remark}
    Given $\alpha$, the most common method for determining the convexity of $F_{\alpha}(\mathbf{X})$ is to check the positive definiteness of its Hessian matrix, since $F_{\alpha}(\mathbf{X})$ is twice continuously differentiable. However, due to the high computational complexity of this approach, an alternative simple method can be employed. Leveraging the fact that the inner function $f_{m}(\mathbf{X})$ is concave according to Lemma \ref{fm_concave}, the convexity of $F_{\alpha}(\mathbf{X})$ can be easily determined using the convexity and monotonicity of the outer function $u_{\alpha}(x)$, together with the rule for composition functions \cite{boyd2004convex}. Specifically, when \(\alpha\) is a rational number, we express \(1-\alpha\) in the form of an irreducible fraction
    \begin{equation}
        1-\alpha = \alpha_{1}^{\mathrm{num}}/ \alpha_{1}^{\mathrm{den}},
    \end{equation}
    where \(\alpha_{1}^{\mathrm{num}}\) and \(\alpha_{1}^{\mathrm{den}}\) are coprime integers. Below, we classify some common cases to facilitate the judgment of convexity:
    \begin{align*}
        \bigl\{ \alpha \;\big|\; \alpha_1^{\mathrm{num}} \text{ and } \alpha_1^{\mathrm{den}} \text{ are both odd} \bigr\} 
        &\subseteq \mathcal{A}_1, \\
        \bigl\{ \alpha \;\big|\; \alpha_1^{\mathrm{num}} \text{ is even and } \alpha_1^{\mathrm{den}} \text{ is odd},\ 
        \alpha \neq 1 \bigr\} 
        &\subseteq \mathcal{A}_2, \\
        \bigl\{ \alpha \;\big|\; \alpha = 1 \text{ or } \alpha \text{ is an irrational number} \bigr\} 
        &\subseteq \mathcal{A}_3, \\
        \bigl\{ \alpha \;\big|\; \alpha \to \infty \bigr\} 
        &\subseteq \mathcal{A}_4.
    \end{align*}
\end{remark}

\begin{remark}
    When directly optimizing convex objectives, the optimal solution to the relaxed problem often fails to satisfy the original discrete constraints. This necessitates a back-projection process \cite{liu2014graph, martello1987linear}, which maps continuous solutions to discrete feasible points under the assumption that the true discrete optimum lies near the continuous one. In practice, however, this proximity assumption frequently breaks down, potentially leading to substantial additional errors in the final results \cite{liu2014graph}.
\end{remark}

\subsection{$Case$ 1: Fast Frank-Wolfe (FFW) Algorithm}\label{case1_FFW}
As mentioned earlier, the FFW algorithm can be leveraged to tackle this problem  effectively and efficiently, owing to the differentiability of its objective function. We can execute the following iterative procedure to implement the FW algorithm \cite{pokutta2024frank}
\begin{subequations}
    \label{FW}
    \begin{align}
        &\mathbf{Y}_{t} = \mathop{\arg \min} _{\mathbf{Y}\in \varOmega} \langle \nabla F_{\alpha}(\mathbf{X}_{t}), \mathbf{Y} \rangle, \label{FW_Yt}\\
        &\mathbf{X}_{t+1} = (1-\upsilon_{t})\mathbf{X}_{t} + \upsilon_{t}\mathbf{Y}_{t},\label{FW_Xt1}
    \end{align}
\end{subequations}
where $\upsilon_{t}$ is step size, defined by
\begin{equation}
    \label{FFW_step_size}
    \upsilon _{t} = \dfrac{2}{2+t}.
\end{equation}
Regarding the linear program \eqref{FW_Yt}, it is precisely a OPT problem\cite{fan2024optimal}. 
Motivated by this structural equivalence, we propose the FFW algorithm, which integrates the classical FW framework with specialized solvers tailored for the OPT problem. These solvers address the subproblem \eqref{FW_Yt} significantly more efficiently than conventional approaches used in the FW method, thereby enabling faster convergence and superior overall performance.
To be specific, we introduce corresponding relaxation variables $\mathbf{z} \in \mathbb{R}_+^{M}$ to account for the inequality constraints $\mathbf{Y} \in \varOmega_{2}$ and entropy regular terms $\sum_{k=1}^K \sum_{m=1}^M y_{km}(\ln y_{km}-1))$ and $\sum_{m=1}^M z_m(\ln z_m -1)$, thus \eqref{FW_Yt} can be transformed by
\begin{subequations}
    \label{problemP3t}
    \begin{align}
        \mathop{\min_{\mathbf{Y}, \mathbf{z}}}&\quad \langle \nabla F_{\alpha}(\mathbf{X}_{t}), \mathbf{Y} \rangle + \sum_{k=1}^K \sum_{m=1}^M \delta y_{km}(\ln y_{km}-1)\nonumber \\
        & \quad \quad  \quad \quad \quad  \quad  +\sum_{m=1}^M \delta z_m(\ln z_m -1) \label{problemP3ta} \\
        \mbox{s.t.} 
        & \quad \sum_{k=1}^{K}y_{km} -z_m= 1,\ \forall m, \label{problemP3tb} \\
        & \quad \sum_{m=1}^{M}y_{km}=1,\ \forall k,\label{problemP3tc}
    \end{align}
\end{subequations}
where $\delta > 0$ is the regularization parameter. The Lagrangian of problem \eqref{problemP3t} is \eqref{Lagrangian_dual}, with dual variables $\boldsymbol{\tilde{\alpha}} \in \mathbb{R}^{K}$ and $\boldsymbol{\tilde{\beta}} \in \mathbb{R}^{M}$. 
By KKT conditions \cite{kuhn2013nonlinear}, we can obtain
    \begin{align*}
        & \frac{\partial \mathcal{L}_t}{\partial y_{km}}=[\nabla F_{\alpha}(\mathbf{X}_{t})]_{km}+\delta \ln y_{km}+\tilde{\alpha}_k+\tilde{\beta}_m=0,\ \forall k,m ,\\
        & \frac{\partial \mathcal{L}_t}{\partial z_{m}}=\delta \ln z_m-\tilde{\beta}_m=0,\ \forall m.
    \end{align*}
Therefore, the optimal solution satisfies
    \begin{align*}
        &  y_{km}=e^{-\frac{[\nabla F_{\alpha}(\mathbf{X}_{t})]_{km}}{\delta}}e^{-\frac{\tilde{\alpha}_k}{\delta}}e^{-\frac{\tilde{\beta}_m}{\delta}},\ \forall k,m ,\\
        &  z_{m}=e^{\frac{\tilde{\beta}_m}{\delta}},\ \forall m.
    \end{align*}
Define auxiliary variables $A_{km}=e^{-\frac{[\nabla F_{\alpha}(\mathbf{X}_{t})]_{km}}{\delta}}, \phi_k=e^{-\frac{\tilde{\alpha}_k}{\delta}}$ and $\psi_m=e^{-\frac{\tilde{\beta}_m}{\delta}}$, and substitute them into \eqref{problemP3tb} and \eqref{problemP3tc}, we can obtain
\begin{subequations}
    \begin{align}
        & \sum_{k=1}^{K}\phi_k A_{km} \psi_m -\psi_m^{-1}- 1=0,\ \forall m,\\
        & \sum_{m=1}^{M}\phi_k A_{km} \psi_m - 1=0,\ \forall k.
    \end{align}
\end{subequations}

Note that $\mathbf{\phi}$ and $\mathbf{\psi}$ can be derived by exploring the idea behind the Sinkhorn algorithm \cite{cuturi2013sinkhorn}, utilizing the iterative procedure outlined below
\begin{equation}
    \begin{aligned}
        \label{psi_phi}
        & \psi_m =\frac{1+\sqrt{1+4\sum_{k=1}^{K} \phi_k A_{km} }}{2\sum_{k=1}^{K} \phi_k A_{km}},\ \forall m,\\
        & \phi_k  =\frac{1}{\sum_{m=1}^{M}A_{km} \psi_m},\ \forall k.
    \end{aligned}
\end{equation}

\begin{figure*}[htbp]
\begin{equation}
    \label{Lagrangian_dual}
    \begin{split}
        \mathcal{L}_{t}(\mathbf{Y},\mathbf{z},\boldsymbol{\tilde{\alpha}},\boldsymbol{\tilde{\beta}})= \langle \nabla F_{\alpha}(\mathbf{X}_{t}), \mathbf{Y} \rangle + \sum_{k=1}^K \sum_{m=1}^M \delta y_{km}(\ln &y_{km}-1) + \sum_{m=1}^M \delta z_m(\ln z_m -1)\\
    &+\sum_{k=1}^K\tilde{\alpha}_{k} \left(\sum_{m=1}^M y_{km}-1\right)+\sum_{m=1}^M\tilde{\beta}_{m} \left(\sum_{k=1}^K y_{km}-z_m-1\right).
    \end{split}
\end{equation}
\hrulefill
\end{figure*}

The details of FFW are summarized in Algorithm \ref{algorithm1}.
\begin{figure}[t]
    \renewcommand{\algorithmicrequire}{\textbf{Input:}}
    \renewcommand{\algorithmicensure}{\textbf{Output:}}
	\begin{algorithm}[H]
		\caption{FFW-based $\alpha$-Fair Clustered Cell-Free Networking}
		\label{algorithm1}
		\begin{algorithmic}[1]
			\REQUIRE Large-scale fading matrix $\mathbf{Q}=[q_{lk}]\in \mathbb{R}^{L\times K}$, the number of users $K$, BSs $L$, and subnetworks $M$, the regularization parameter $\delta > 0$ and the iteration numbers $T_{F}$, $S_{F}$ and $N_{F}$.
			\ENSURE Network decomposition $\mathcal{M}$.
			\STATE{\textbf{Initialization:} $\mathbf{X}_{0}\in \varOmega$.}
                \STATE{Group BSs to obtain decomposition $\mathcal{B} = \{\mathcal{B}_{1}, \mathcal{B}_{2}, ..., \mathcal{B}_{M}\}$.}
			\FOR{$t = 0, ..., T_{F}$, }
                \STATE{Compute $\boldsymbol{\alpha}^{m}$ and $\boldsymbol{\beta}^{m}$ by $N_{F}$ alternating iterations with \eqref{alpha_beta}.} \label{algorithm1_compute_alpha_beta}
                \STATE{Compute $\nabla F_{\alpha}(\mathbf{X}_{t})$ by \eqref{Fderivative} and $A_{km}=e^{-\frac{[\nabla F_{\alpha}(\mathbf{X}_{t})]_{km}}{\delta}}$.}
                \FOR{$s = 0, ..., S_{F}$, }
                \STATE{Compute $\psi_m$ and $\phi_m$ by \eqref{psi_phi}.} \label{algorithm_compute_psi_phi}
                \ENDFOR
                \STATE{Compute $\mathbf{Y}_{t}$ by $y_{km}=\phi_k A_{km} \psi_m$.}
                \STATE{Compute step size $\upsilon _{t} = 2/(2+t)$ by \eqref{FFW_step_size}.}
                \STATE{Update $\mathbf{X}_{t+1} = (1-\upsilon_{t})\mathbf{X}_{t} + \upsilon_{t}\mathbf{Y}_{t}$ by \eqref{FW_Xt1}.}
			\ENDFOR
                \STATE Let $\mathbf{X}^{*} = \mathbf{X}_{T_{B}}$.
                \FOR{$k = 0, ..., K$, }
			\STATE{Let $m^{*} = \mathrm{arg}\max_{1\leq m \leq M}x^{*}_{km}$ and assign $u_{k}$ to $\mathcal{S}_{m}$.}
                \ENDFOR
                \RETURN{Network decomposition $\mathcal{M}= \{\mathcal{S}_{1}, \mathcal{S}_{2}, ..., \mathcal{S}_{M}\}$}.
		\end{algorithmic}
	\end{algorithm}
\end{figure}
Balanced $k$-means is adopted to obtain the BS partition, which has the computational complexity of $\mathcal{O}(L^{3})$ \cite{malinen2014balanced}. During the iteration, we need $\mathcal{O}(N_{F}KL)$ time to iteratively compute $\boldsymbol{\alpha}^{m}$ and $\boldsymbol{\beta}^{m}, \forall m$ in line \ref{algorithm1_compute_alpha_beta} and $\mathcal{O}(S_{F}KM)$ time to iteratively solve $\psi_m, \forall m$ and $\phi_k, \forall k$ in line \ref{algorithm_compute_psi_phi}. Therefore, the main computational cost during the iteration is $\mathcal{O}(T_{F}(N_{F}KL + S_{F}KM))$ = $\mathcal{O}(KL)$. The convergence behavior of the Frank-Wolfe (FW) algorithm has been rigorously established in \cite{pokutta2024frank}, where the iterative structure of the convergence proof closely aligns with equations \eqref{FW} and \eqref{FFW_step_size}, with the only difference being that FFW formulates \eqref{FW_Yt} as an OPT problem. However, the OPT-enhanced component only accelerates convergence without affecting the final solution of \eqref{FW_Yt}, leading to the convergence of the FFW algorithm presented in Algorithm \ref{algorithm1} is similarly guaranteed.


\subsection{$Case$ 2: Convex-Concave Relaxation Procedure with FFW (CCRP-FFW) algorithm}
The equivalence of the relaxed problem $\mathcal{P}3$ to $\mathcal{P}2$ is primarily due to the concavity of $\mathcal{P}3$'s objective function $F_{\alpha}(\mathbf{X})$. Therefore, for case 2 where $F_{\alpha}(\mathbf{X})$ is convex, this subsection explores the CCRP approach \cite{liu2014graph}, which facilitates a gradual transition of the objective function in $\mathcal{P}2$ from a convex relaxation to a concave relaxation, with its minima locating exactly in the original discrete domain. 

Consider the following auxiliary function 
\begin{equation}
    \label{convex_concave_gamma}
    F_{\gamma}^{\mathrm{aux}}(\mathbf{X}) = (1-\gamma)\widetilde{F}_{\alpha}(\mathbf{X}) + \gamma \overline{F}_{\alpha}(\mathbf{X}),
\end{equation}
where $\gamma \in [0, 1]$ is a parameter, $\widetilde{F}_{\alpha}(\mathbf{X})$ is a convex function and $\overline{F}_{\alpha}(\mathbf{X})$ is a concave function if $\mathbf{X} \in \varOmega$. Consequently, the auxiliary function $F_{\gamma}^{\mathrm{aux}}(\mathbf{X})$ evolves from a convex to a concave nature as the parameter $\gamma$ increases from 0 to 1. 

The essence of CCRP is in effectively utilizing ${F}_{\alpha}(\mathbf{X})$ to determine the appropriate $\widetilde{F}_{\alpha}(\mathbf{X})$ and $\overline{F}_{\alpha}(\mathbf{X})$. Note that when $\alpha_{1}^{\mathrm{num}}$ is even and $\alpha_{1}^{\mathrm{den}}$ is odd ($\alpha \neq 1$), $F_{\alpha}(\mathbf{X})$ is a convex function, so we can simply set
\begin{equation}
    \label{convex_relaxation}
    \widetilde{F}_{\alpha}(\mathbf{X}) = F_{\alpha}(\mathbf{X}).
\end{equation}

Regarding the construction of $\overline{F}_{\alpha}(\mathbf{X})$, 
CCRP proposes a simple and universally applicable approach for constructing a general concave relaxation, suitable for a wide range of problems with convex objective function, thereby enabling its effective realization for $\overline{F}_{\alpha}(\mathbf{X})$. Specifically, To obtain $\overline{F}_{\alpha}(\mathbf{X})$ with ${F}_{\alpha}(\mathbf{X})$, consider the following formula
\begin{equation}
    \label{concave_relaxation_derivation}
    \begin{aligned}
    & \ \mathop{\arg \min}_{\mathbf{X}\in \varOmega_{01}} \ {F}_{\alpha}(\mathbf{X})  \\
    =& \ \mathop{\arg \min}_{\mathbf{X}\in\varOmega_{01}} \ [(1-\eta){F}_{\alpha}(\mathbf{X}) - \eta K] \\
    =& \ \mathop{\arg \min}_{\mathbf{X}\in\varOmega_{01}} \ [(1-\eta){F}_{\alpha}(\mathbf{X}) - \eta\mathrm{tr}(\mathbf{X}^{T}\mathbf{X})],
    \end{aligned}
\end{equation}
where $\eta \in [0, 1]$ is a parameter and $\mathbf{X} \in \varOmega_{01}$ is the constraints of $\mathcal{P}2$, corresponding to the discrete case of $\mathbf{X} \in \varOmega$, and can be formulated as
\begin{equation}
\begin{split}
    \varOmega_{01} = \left\{\mathbf{X}\in \{0, 1\}^{K\times M} \bigg| \sum_{k=1}^{K} \right.& x_{km} \geq 1, \forall m, \\
    &\left.\sum_{m=1}^{M}x_{km} = 1, \forall k \right\}.
\end{split}
\end{equation}
Note that 
$K = \mathrm{tr}(\mathbf{X}^{T}\mathbf{X})$ in \eqref{concave_relaxation_derivation} because of $\sum_{m=1}^{M}x_{km} = 1, \forall k$.

We aim to identify the appropriate range of $\eta$ that will render $(1-\eta){F}_{\alpha}(\mathbf{X}) - \eta\mathrm{tr}(\mathbf{X}^{T}\mathbf{X})$ a concave function when $\mathbf{X} \in \varOmega$. Compute its Hessian matrix $H(\mathbf{X})$
\begin{equation}
    H(\mathbf{X}) = (1-\eta)H_{\alpha}(\mathbf{X}) - 2\eta \mathbf{E},
\end{equation}
where $H_{\alpha}(\mathbf{X})$ is the Hessian matrix of $F_{\alpha}(\mathbf{X})$ and $\mathbf{E} \in \mathbb{R}^{KM\times KM}$ is an identity matrix. To obtain concave nature, $\eta$ should satisfy
\begin{equation}
    \label{eta_range}
    1 \geq \eta > \dfrac{\lambda_{\mathrm{max}}}{2 + \lambda_{\mathrm{max}}},
\end{equation}
where $\lambda_{\mathrm{max}}$ denotes the maximal eigenvalue of $H_{\alpha}(\mathbf{X})$. Thus, we can set 
\begin{equation}
    \label{concave_relaxation}
    \overline{F}_{\alpha}(\mathbf{X}) = (1-\eta){F}_{\alpha}(\mathbf{X}) - \eta\mathrm{tr}(\mathbf{X}^{T}\mathbf{X}),
\end{equation}
with $\eta$ meets \eqref{eta_range}. 

The auxiliary function \eqref{convex_concave_gamma} can be transformed by
\begin{equation}
    \label{auxiliary_function_transformed}
    F_{\gamma}^{\mathrm{aux}}(\mathbf{X})
    = (1-\zeta)F_{\alpha}(\mathbf{X}) - \zeta \mathrm{tr}(\mathbf{X}^{T}\mathbf{X})
    \triangleq F_{\zeta}^{\mathrm{aux}}(\mathbf{X}),
\end{equation}
where $\zeta \triangleq \gamma\eta \in [0, \eta]$. We can transform $\mathcal{P}2$ into
\begin{subequations}
    \begin{align}
        \min_{\mathbf{X}} & \quad F_{\zeta}^{\mathrm{aux}}(\mathbf{X})&& \label{problemP4a}\\
        \mbox{s.t.} 
        & \quad \mathbf{X} \in \varOmega_{01}, \label{problemP4b} 
    \end{align}
\end{subequations}
which is equivalent to its relaxed problem with $\mathbf{X} \in \varOmega$. While the existence of $\zeta$ that renders $F_{\zeta}^{\mathrm{aux}}(\mathbf{X})$ concave is theoretically guaranteed, identifying it in reality demands significant computational effort. Consequently, we employ an iterative method to solve $\mathcal{P}4_{(t)}$ with gradually increasing $\zeta$ from 0 to 1 until the local optimal solution $\mathcal{P}2$ is identified.
\begin{subequations}
    \begin{align}
        \mathcal{P}4_{(t)}:
        \min_{\mathbf{X}} & \quad F_{\zeta_{t}}^{\mathrm{aux}}(\mathbf{X})&& \label{problemP4at}\\
        \mbox{s.t.} 
        & \quad \mathbf{X} \in \varOmega, \label{problemP4bt} 
    \end{align}
\end{subequations}
which can be solved by FFW algorithm in Algorithm \ref{algorithm1}. 

Since CCRP is essentially an iterative wrapper around the FFW algorithm, and given that FFW algorithm itself has a complexity of $\mathcal{O}(T_{F}(N_{F}KL + S_{F}KM))$ = $\mathcal{O}(KL)$ analyzed in the end part of Section \ref{case1_FFW}, the overall computational complexity of the CCRP approach presented in Algorithm \ref{algorithm2} becomes $\mathcal{O}(T_{C}T_{F}(N_{F}KL + S_{F}KM))$ = $\mathcal{O}(KL)$.
\begin{figure}[t]
    \renewcommand{\algorithmicrequire}{\textbf{Input:}}
    \renewcommand{\algorithmicensure}{\textbf{Output:}}
	\begin{algorithm}[H]
		\caption{CCRP-based $\alpha$-Fair Clustered Cell-Free Networking}
		\label{algorithm2}
		\begin{algorithmic}[1]
			\REQUIRE Large-scale fading matrix $\mathbf{Q}=[q_{lk}]\in \mathbb{R}^{L\times K}$, the number of users $K$, BSs $L$, and subnetworks $M$, the regularization parameter $\delta > 0$, step size $S_{C} >0$ and the iteration numbers $T_{C}$, $T_{F}$, $S_{F}$ and $N_{F}$.
			\ENSURE Network decomposition $\mathcal{M}$.
			\STATE{\textbf{Initialization:} $\mathbf{X}_{0}\in \varOmega$ and $\zeta_{0} = 0$.}
                \STATE{Group BSs to obtain decomposition $\mathcal{B} = \{\mathcal{B}_{1}, \mathcal{B}_{2}, ..., \mathcal{B}_{M}\}$.}
			\FOR{$t = 0, ..., T_{C}$, }
                \STATE{Update $\zeta_{t} = \zeta_{t-1} + S_{C}$, where $\zeta_{t} \in [0, 1]$.}
                \STATE{Compute $F_{\zeta_{t}}^{\mathrm{aux}}(\mathbf{X})$ by \eqref{auxiliary_function_transformed}.}
                \STATE{Obtain $\mathbf{X}^{*}_{t}$ through FFW algorithm solving $\mathcal{P}4_{(t)}$.}
			\ENDFOR
                \STATE Let $\mathbf{X}^{*} = \mathbf{X}_{T_{C}}$.
                \FOR{$k = 0, ..., K$, }
			\STATE{Let $m^{*} = \mathrm{arg}\max_{1\leq m \leq M}x^{*}_{km}$ and assign $u_{k}$ to $\mathcal{S}_{m}$.}
                \ENDFOR
                \RETURN{Network decomposition $\mathcal{M}= \{\mathcal{S}_{1}, \mathcal{S}_{2}, ..., \mathcal{S}_{M}\}$}.
		\end{algorithmic}
	\end{algorithm}
\end{figure}

\subsection{$Case$ 3: Graduated NonConvexity and Concavity Procedure with FFW (GNCCP-FFW) algorithm}
In the case of $F_{\alpha}(\mathbf{X})$ being neither convex nor concave, CCPR with FFW algorithm is not capable of handling this situation, due to the difficulty involved in identifying the convex relaxation according to \eqref{convex_relaxation}. GNCCP \cite{liu2013gnccp} leverages alternative approach to successfully trackle this challenge. Unlike \eqref{convex_concave_gamma} only requiring iteratively transforming the original objective function from convex to concave, GNCCP employs the following auxiliary function $F_{\gamma}^{\mathrm{aux}}(\mathbf{X})$ to construct convex and concave relaxations, where the iterative procedure transforms the original neither convex nor concave function first into convex and finally into concave. The convexification applied in the first stage converges the solution to the unique optimum of the convex objective function, thereby mitigating the effects of local optima.

Define the auxiliary function $F_{\gamma}^{\mathrm{aux}}(\mathbf{X})$ by
\begin{equation}
    \label{convex_concave_gamma_GNCCP}
    \!\!\!\!F_{\gamma}^{\mathrm{aux}}(\mathbf{X}) \!\triangleq \!
    \begin{cases}
	(1-\gamma)F_{\alpha}(\mathbf{X}) + \gamma \mathrm{tr}(\mathbf{X}^{T}\mathbf{X}), &\!\!\!\! 1 \geq \gamma \geq 0,\\
	(1+\gamma)F_{\alpha}(\mathbf{X}) + \gamma \mathrm{tr}(\mathbf{X}^{T}\mathbf{X}), &\!\!\!\! 0 \geq \gamma \geq -1,
    \end{cases}
\end{equation}
where $\gamma$ is a parameter varying from 1 to -1. Thus \eqref{convex_concave_gamma_GNCCP} can realize a convex and concave relaxation procedure with following convex function $\widetilde{F}_{\alpha}(\mathbf{X})$ and concave function $\overline{F}_{\alpha}(\mathbf{X})$, respectively, given by
\begin{equation}
    \begin{aligned}
        & \widetilde{F}_{\alpha}(\mathbf{X}) = F_{\alpha}(\mathbf{X}) - \lambda_{\mathrm{min}}\mathrm{tr}(\mathbf{X}^{T}\mathbf{X} - \mathbf{1}_{M\times K}\mathbf{X}),  \\
        & \overline{F}_{\alpha}(\mathbf{X}) = F_{\alpha}(\mathbf{X}) - \lambda_{\mathrm{max}}\mathrm{tr}(\mathbf{X}^{T}\mathbf{X} - \mathbf{1}_{M\times K}\mathbf{X}).
    \end{aligned}
\end{equation}
Here $\mathbf{1}_{M\times K} \in \mathbb{R}^{K\times M}$ is the unit matrix with each entry being 1, $\lambda_{\mathrm{min}}$ and $\lambda_{\mathrm{max}}$ denote the minimal and maximal eigenvlaues of the Hessian matrix of $F_{\alpha}(\mathbf{X})$, respectively.

With $F_{\gamma}^{\mathrm{aux}}(\mathbf{X})$ defined by \eqref{convex_concave_gamma_GNCCP}, we can iteratively solve following $\mathcal{P}5_{(t)}$ according to $\mathcal{P}4_{(t)}$ to obtain the local optimal solution of $\mathcal{P}2$.
\begin{subequations}
    \begin{align}
        \mathcal{P}5_{(t)}:
        \min_{\mathbf{X}} & \quad F_{\gamma_{t}}^{\mathrm{aux}}(\mathbf{X})&& \label{problemP5at}\\
        \mbox{s.t.} 
        & \quad \mathbf{X} \in \varOmega. \label{problemP5bt} 
    \end{align}
\end{subequations}
The details of GNCCP algorithm are summarized in Algorithm \ref{algorithm3}. Similar to CCRP algorithm, the GNCCP approach, nests $T_{G}$ iterations over the FFW algorithm, resulting in a time complexity of $\mathcal{O}(T_{G}T_{F}(N_{F}KL + S_{F}KM))$ = $\mathcal{O}(KL)$ for GNCCP.
\begin{figure}[t]
    \renewcommand{\algorithmicrequire}{\textbf{Input:}}
    \renewcommand{\algorithmicensure}{\textbf{Output:}}
	\begin{algorithm}[H]
		\caption{GNCCP-based $\alpha$-Fair Clustered Cell-Free Networking}
		\label{algorithm3}
		\begin{algorithmic}[1]
			\REQUIRE Large-scale fading matrix $\mathbf{Q}=[q_{lk}]\in \mathbb{R}^{L\times K}$, the number of users $K$, BSs $L$, and subnetworks $M$, the regularization parameter $\delta > 0$, step size $S_{G}>0$ and the iteration numbers $T_{G}$, $T_{F}$, $S_{F}$ and $N_{F}$.
			\ENSURE Network decomposition $\mathcal{M}$.
			\STATE{\textbf{Initialization:} $\mathbf{X}_{0}\in \varOmega$ and $\gamma_{0} = 1$.}
                \STATE{Group BSs to obtain decomposition $\mathcal{B} = \{\mathcal{B}_{1}, \mathcal{B}_{2}, ..., \mathcal{B}_{M}\}$.}
			\FOR{$t = 0, ..., T_{G}$, }
                \STATE{Update $\gamma_{t} = \gamma_{t-1} - S_{G}$, where $\zeta_{t} \in [-1, 1]$.}
                \STATE{Compute $F_{\gamma_{t}}^{\mathrm{aux}}(\mathbf{X})$ by \eqref{convex_concave_gamma_GNCCP}.}
                \STATE{Obtain $\mathbf{X}^{*}_{t}$ through FFW algorithm solving $\mathcal{P}5_{(t)}$.}
			\ENDFOR
                \STATE Let $\mathbf{X}^{*} = \mathbf{X}_{T_{C}}$.
                \FOR{$k = 0, ..., K$, }
			\STATE{Let $m^{*} = \mathrm{arg}\max_{1\leq m \leq M}x^{*}_{km}$ and assign $u_{k}$ to $\mathcal{S}_{m}$.}
                \ENDFOR
                \RETURN{Network decomposition $\mathcal{M}= \{\mathcal{S}_{1}, \mathcal{S}_{2}, ..., \mathcal{S}_{M}\}$}.
		\end{algorithmic}
	\end{algorithm}
\end{figure}

\subsection{$Case$ 4: Alternating Gradient Projection (AGP) Algorithm}\label{AGP}
To solve $\mathcal{P}2$ with $\alpha \rightarrow \infty$, we similarly employ a relaxation methodology to equivalently transform it into a continuous optimization problem. Specifically, relax $\mathcal{P}2$ with $\alpha \rightarrow \infty$ to the following min-max problem $\mathcal{P}6$, where the equivalence rigorously established by Theorem \ref{equivalent_P2_P6}.
\begin{equation}
    \mathcal{P}6: \min_{\mathbf{X}\in \varOmega}  \max_{\mathbf{y} \in \Delta}  f(\mathbf{X}, \mathbf{y}),
\end{equation}
where $f(\mathbf{X}, \mathbf{y}) = \sum_{m=1}^{M}y_{m}f_{m}(\mathbf{X})$ with $\Delta = \{\mathbf{y}\in \mathbb{R}^{M} | \mathbf{y} \geq 0; \sum_{m=1}^{M}y_{m} = 1 \}$.

Since a min-max problem can be solved by finding a Nash equilibrium (or saddle point) \cite{nouiehed2019solving}, let us first present the definition of Nash equilibrium, and then demonstrate the tightness of the relaxation of $\mathcal{P}2$ into $\mathcal{P}6$ in Theorem \ref{equivalent_P2_P6}.

\begin{definition}[Nash equilibrium\cite{nouiehed2019solving}]
    Let $\mathcal{X}$ and $\mathcal{Y}$ be two convex sets. We say that $(x^{*}, y^{*}) \in \mathcal{X}\times \mathcal{Y}$ is a Nash equilibrium (or saddle point) of $\min_{x\in \mathcal{X}}  \max_{y \in \mathcal{Y}}  f(x, y)$ if 
    \begin{equation}
        \label{Nash_equilibrium}
        f(x^{*}, y) \leq f(x^{*}, y^{*}) \leq f(x, y^{*}), \forall x\in \mathcal{X}, \forall y\in \mathcal{Y}.
    \end{equation}
\end{definition}

\begin{thm}\label{equivalent_P2_P6}
    If $\mathcal{P}6$ has a Nash equilibrium, then $\mathcal{P}2$ and $\mathcal{P}6$ have the same optimal value and share a common optimal solution.
\end{thm}
\begin{IEEEproof}
    See Appendix \ref{equivalent_P2_P6_proof}.
\end{IEEEproof}

According to Theorem \ref{equivalent_P2_P6}, solving $\mathcal{P}2$ reduces entirely to identifying the Nash equilibrium in $\mathcal{P}6$. However, the objective function $f(\mathbf{X}, \mathbf{y})$ in $\mathcal{P}6$ is concave with respect to $\mathbf{X}$ according to Lemma \ref{fm_concave} and linear with respect to $\mathbf{y}$, resulting in $\mathcal{P}3$ being a nonconvex-linear min-max problem, which is known to be challenging to solve. 
Recently, an efficient alternating gradient projection (AGP) algorithm is designed in \cite{pan2021efficient} to tackle nonconvex-linear min-max problems by exploiting Danskin’s Theorem \cite{Danskin} and the property that $f(\mathbf{X}, \mathbf{y})$ is linear with respect to $\mathbf{y}$. 

By applying the AGP algorithm to solve our problem $\mathcal{P}6$, $\mathbf{X}$ and $\mathbf{y}$ are alternately updated in the $t$-th iteration by \cite{pan2021efficient}
\begin{subequations}
    \label{AGP_iteration}
    \begin{align}
        &\mathbf{y}_{t+1} = \mathop{\arg\max}_{\mathbf{y}\in \Delta} f_{\lambda, \gamma_{t}}(\mathbf{X}_{t}, \mathbf{y}), \label{AGP_iteration_y}\\
        &\mathbf{X}_{t+1} = P_{\varOmega}(\mathbf{X}_{t} - \zeta_{t}\nabla_{\mathbf{X}}f(\mathbf{X}_{t}, \mathbf{y}_{t+1})),\label{AGP_iteration_X}
    \end{align}
\end{subequations}
where $P_{\varOmega}(\mathbf{X})$ denotes the projection of $\mathbf{X}$ onto $\varOmega$ and $\zeta_{t}$ is the step size in the $t$-iteration, which will be specifically optimized in subsection \ref{PC}. The elements of the gradient $\nabla_{\mathbf{X}}f(\mathbf{X}, \mathbf{y})$ can be found from \eqref{problemP2a} and \eqref{sum_capacity_deterministic_equivalents}--\eqref{alpha_beta} as
\begin{equation}
    \label{f_diff}
    \dfrac{\partial f}{\partial x_{km}} = -\sum_{b_{l}\in \mathcal{B}_{m}}\beta_{k}^{m}\theta_{lk}\alpha_{l}^{m}y_{m}.
\end{equation}
Besides, $f_{\lambda, \gamma_{t}}(\mathbf{X}, \mathbf{y})$ in \eqref{AGP_iteration_y} is an auxiliary function introduced to replace the original objective function $f(\mathbf{X}, \mathbf{y})$ such that the objective function is strongly concave with respect to $\mathbf{y}$ for any fixed $\mathbf{X}$ and the gradient of $\max_{\mathbf{y}\in\Delta}f_{\lambda, \gamma_{t}}(\mathbf{X}, \mathbf{y})$ with respect to $\mathbf{X}$ can be obtained by the following Danskin’s Theorem. 

\begin{thm}[Danskin’s Theorem \cite{Danskin}] \label{Danskin_theorem}
Assume $f(x, y)$ is differentiable in $x \in \mathcal{X}$ for every $y \in \mathcal{Y}$ and strongly concave in $y \in \mathcal{Y}$ for every $x \in \mathcal{X}$ with $\mathcal{Y}$ being compact. Then, function $g(x) \triangleq \max_{y\in \mathcal{Y}} f(x, y)$ is differentiable in $x$ and for any 
$x\in \mathcal{X}$, the gradient $\nabla g(x) = \nabla_{x} f(x, y^{*})$, where $y^{*} = \arg \max_{y\in \mathcal{Y}}f(x, y)$.
\end{thm}

Specifically, $f_{\lambda, \gamma_{t}}(\mathbf{X}, \mathbf{y})$
in the $t$-th iteration is defined as
\begin{equation}
    \hspace{-1em}
    \label{auxiliary_f_X_y}
    f_{\lambda, \gamma_{t}}(\mathbf{X}, \mathbf{y}) \triangleq f(\mathbf{X}, \mathbf{y}) - \dfrac{\lambda}{2}\|\mathbf{y} - \overline{\mathbf{y}}\|_{2}^{2} - \dfrac{\gamma_{t}}{2}\|\mathbf{y} - \mathbf{y}_{t}\|_{2}^{2},
\end{equation}
where $\overline{\mathbf{y}}$ is a fixed given point in set $\Delta$, $\lambda$ and $\gamma_{t}$ are two regularization parameters that will be configured in subsection \ref{PC}. It can be seen that $f_{\lambda, \gamma_{t}}(\mathbf{X}, \mathbf{y})$ is a $(\lambda + \gamma_{t})$-strongly concave quadratic function with respect to $\mathbf{y}$ for any fixed $\mathbf{X}$, which meets the conditions of Danskin’s Theorem. Therefore, the gradient of $\max_{\mathbf{y}\in\Delta}f_{\lambda, \gamma_{t}}(\mathbf{X}, \mathbf{y})$ with respect to $\mathbf{X}$ at $\mathbf{X} = \mathbf{X}_{t}$ can be obtained as $ \nabla_{\mathbf{X}}f_{\lambda, \gamma_{t}}(\mathbf{X}_{t}, \mathbf{y}_{t+1})$ due to \eqref{AGP_iteration_y}. 
As the difference between $f_{\lambda, \gamma_{t}}(\mathbf{X}_{t}, \mathbf{y})$ and $f(\mathbf{X}_{t}, \mathbf{y})$ is just two quadratic terms with respect to $\mathbf{y}$, it is obvious that $\nabla_{\mathbf{X}}f_{\lambda, \gamma_{t}}(\mathbf{X}_{t}, \mathbf{y})=\nabla_{\mathbf{X}}f(\mathbf{X}_{t}, \mathbf{y})$ for any given $\mathbf{y} \in \Delta$. That is, \eqref{AGP_iteration} is indeed an alternating iterative process regarding $f_{\lambda, \gamma_{t}}(\mathbf{X}, \mathbf{y})$, where the gradient projection is performed on variable $\mathbf{X}$.

\subsubsection{Iteratively solving for $\mathbf{y}_{t+1}$}
According to \cite{9186144}, to ensure that the iterative process outlined in \eqref{AGP_iteration} correctly identifies the solution to $\mathcal{P}6$, sub-problem \eqref{AGP_iteration_y} must be precisely solved. Thanks to the linearity of $f(\mathbf{X}_{t}, \mathbf{y})$ concerning variable $\mathbf{y}$, $f_{\lambda, \gamma_{t}}(\mathbf{X}_{t}, \mathbf{y})$ is a quadratic function with respect to $\mathbf{y}$, and therefore can be exactly solved. Specifically, finding the optimal solution $\mathbf{y}_{t+1}$ in \eqref{AGP_iteration_y} is equivalent to solving the following optimaztion problem
\begin{equation}
    \label{min_y}
    \min_{\mathbf{y}\in \Delta} \|\mathbf{y} - \dfrac{1}{\lambda+\gamma_{t}}(\lambda\overline{\mathbf{y}} + \gamma_{t}\mathbf{y}_{t} + \mathbf{f}(\mathbf{X}_{t}))\|_{2}^{2},
\end{equation}
where $\mathbf{f}(\mathbf{X}_{t}) = [f_{1}(\mathbf{X}_{t}), f_{2}(\mathbf{X}_{t}), ..., f_{M}(\mathbf{X}_{t})]^{T}$. For simplify, let $\boldsymbol{a} \triangleq (\lambda\overline{\mathbf{y}} + \gamma_{t}\mathbf{y}_{t} + \mathbf{f}(\mathbf{X}_{t}))/(\lambda+\gamma_{t}) \in \mathbb{R}^{M}$. We can obtain the optimal solution $\mathbf{y}_{t+1}$ within finite iterations using the following Lemma \ref{find_y_t}.

\begin{lemma}\label{find_y_t}
    For optimization problem $\mathrm{arg} \min_{\mathbf{y}\in \Delta} \|\mathbf{y} - \boldsymbol{a}\|^{2}_{2}$ with $\boldsymbol{a} \in \mathbb{R}^{M}$, there exist an integer $N \leq M$, and an optimal solution $\mathbf{y}^{*} \in \Delta$ such that\footnote{Assume that the elements of $\boldsymbol{a}$ are sorted in descending order.}
    \begin{equation}
        \label{find_N}
        y^{*}_{m} = 
        \begin{cases}
        a_{m} - \dfrac{1}{N}\left(\sum_{n=1}^{N}a_{n} -1\right), & m = 1, \cdots, N,\\
        0, & \mathrm{otherwise}.
        \end{cases}
    \end{equation}
\end{lemma}
\begin{IEEEproof}
    See Appendix \ref{find_y_t_proof}.
\end{IEEEproof}

\subsubsection{Iteratively solving for $\mathbf{X}_{t+1}$}
To obtain the projection $P_{\varOmega}(\mathbf{X}_{t} - \zeta_{t}\nabla_{\mathbf{X}}f(\mathbf{X}_{t}, \mathbf{y}_{t+1}))$, Dykstra’s algorithm \cite{boyle1986method} is employed in this paper, whose global convergence can be found in \cite{bauschke2000dykstras}. Dykstra’s algorithm starts by initializing
\begin{equation}
    \begin{aligned}
    &\mathbf{X}^{(0)}_{t+1} = \mathbf{X}_{t} - \zeta_{t}\nabla_{\mathbf{X}}f(\mathbf{X}_{t}, \mathbf{y}_{t+1}),\\
    &\mathbf{Z}^{(0)} = \mathbf{Z}^{(-1)} = \mathbf{Z}^{(-2)} = \mathbf{1}_{K\times M},
    \end{aligned}
\end{equation}
and the sequences iterate via
\begin{equation}
    \label{Dykstra}
    \mathbf{X}^{(s)}_{t+1} {=} P_{\varOmega_{s}}(\mathbf{X}^{(s-1)}_{t+1} \circ \mathbf{Z}^{(s-3)}),~
    \mathbf{Z}^{(s)} \!= \mathbf{Z}^{(s-3)} \circ \dfrac{\mathbf{X}^{(s-1)}_{t+1}}{\mathbf{X}^{(s)}_{t+1}},
\end{equation}
with $\varOmega_{s+3} = \varOmega_{s}, \forall s\in \mathbb{N}^{*}$, where $\mathbf{1}_{K\times M} \in \mathbb{R}^{K\times M}$ is an all-ones matrix, $\frac{\cdot}{\cdot}$ denotes the entrywise division and $\circ$ denotes the Hadamard product. Here, we employ the Kullback-Leibler (KL) projection such that the explicit form of the projection can be obtained by KKT conditions \cite{kuhn2013nonlinear}. Specifically, we have
\begin{subequations}
    \label{Omega_all}
    \begin{align}
    &P_{\varOmega_{1}}(\mathbf{X}) = \max (\mathbf{X}, \mathbf{0}_{K\times M}), \label{Omega1_KL}  \\
    &P_{\varOmega_{2}}(\mathbf{X}) = \mathbf{X} \mathrm{diag} \left( \max \left( \dfrac{\boldsymbol{1}_{M}}{\mathbf{X}^{T}\boldsymbol{1}_{K}}, \boldsymbol{1}_{M} \right)\right), \label{Omega2_KL}\\
    &P_{\varOmega_{3}}(\mathbf{X}) = \mathrm{diag} \left(  \dfrac{\boldsymbol{1}_{K}}{\mathbf{X}\boldsymbol{1}_{M}} \right) \mathbf{X}, \label{Omega3_KL}
    \end{align}
\end{subequations}
where $\max(\cdot, \cdot)$ is the entrywise operation and $\mathrm{diag}(\mathbf{v})$ generates a diagonal matrix using vector $\mathbf{v}$. Here, $\boldsymbol{1}_{M} \in \mathbb{R}^{M}$ is an all-ones column vector and $\mathbf{0}_{K\times M} \in \mathbb{R}^{K\times M}$ is a zero matrix.

\subsubsection{Parameter configuration}\label{PC}
Note that the parameters involved in the iterative process given in \eqref{AGP_iteration}, i.e., $\lambda$, $\gamma_{t}$ and $\zeta_{t}$, determine the convergence of the algorithm and also the convergence speed. 

For the two feasible sets $\varOmega$ and $\Delta$ defined below $\mathcal{P}3$ and $\mathcal{P}3$, there are two balls containing each of them. Let us denote the radius of the larger ball as $R$. The parameters can be set based on \cite{pan2021efficient} for convergence as
\begin{subequations}
    \label{parameter_configuration}
    \begin{align}
    &\lambda = \dfrac{\varepsilon}{8R^{2}} \label{lambda},  \\
    &\gamma_{t} = \dfrac{\mu_{t}}{(t+1)^{\rho}},\ \mu_{t} \leq \mu,\ \rho > 1, \label{gamma}\\
    &\zeta_{t} = \left( \dfrac{\kappa_{11}}{2} + \dfrac{\kappa_{12}^{2}}{\lambda + \gamma_{t}}\right)^{-1} \label{zeta},
    \end{align}
\end{subequations}
where $\varepsilon$ is the accuracy of the $\varepsilon$-first-order Nash equilibrium ($\varepsilon$-FNE)\footnote{$\varepsilon$-FNE serves as a criterion for finding the Nash equilibrium based on the first-order optimality measure.} \cite{nouiehed2019solving}, $\mu_{t}$ is a bounded function and $\rho$ is a constant. As for $\kappa_{11}$ and $\kappa_{12}$, they are Lipschitz constants of $\nabla_{\mathbf{X}}f(\mathbf{X}, \mathbf{y})$ satisfying
\begin{subequations}
    \label{lp}
    \begin{align}
    &\!\!\! \|\nabla_{\mathbf{X}}f(\mathbf{X}_{1}, \mathbf{y}) - \nabla_{\mathbf{X}}f(\mathbf{X}_{2}, \mathbf{y})\|_{2} \leq \kappa_{11} \|\mathbf{X}_{1} - \mathbf{X}_{2}\|_{2}, \label{lp1}\\
    &\!\!\! \|\nabla_{\mathbf{X}}f(\mathbf{X}, \mathbf{y}_{1}) - \nabla_{\mathbf{X}}f(\mathbf{X}, \mathbf{y}_{2})\|_{2} \leq \kappa_{12} \|\mathbf{y}_{1} - \mathbf{y}_{2}\|_{2}. \label{lp2}
    \end{align}
\end{subequations}
With \eqref{lp1} and \eqref{lp2}, we can find that $\nabla_{\mathbf{X}}f(\mathbf{X}, \mathbf{y})$ is bounded, which is a requirement for the convergence of the algorithm.

In order to accelerate the convergence of the AGP algorithm, we need to find the maximum step size $\zeta_{t}$ for each iteration $t$ under the conditions of \eqref{lp1} and \eqref{lp2}. This is equivalent to minimizing $\kappa_{11}$ and $\kappa_{12}$. Let us first introduce the below lemma.
\begin{lemma}
    \label{Hadamard_Cauchy}
    Let $\mathbf{A}, \mathbf{B}$ be $n\times m$ matrices $(n \geq m)$, there exists a constant $\xi$ such that
    \begin{equation}
    \label{Hadamard_Cauchy_xi}
    \|\mathbf{A} \circ \mathbf{B}\|_{2} \leq \xi\max_{i,j}|\mathbf{A}(i,j)|\|\mathbf{B}\|_{2}.
    \end{equation}
    Furthermore, the minimum $\xi$ satisfying \eqref{Hadamard_Cauchy_xi} for any $\mathbf{A}$ and $\mathbf{B}$ is $\sqrt{m}$.
\end{lemma}
\begin{IEEEproof}
    See Appendix \ref{Hadamard_Cauchy_proof}.
\end{IEEEproof}

Since $f(\mathbf{X}, \mathbf{y})$ is linear with respect to $\mathbf{y}$, noting the definition of $\nabla_{\mathbf{X}}f(\mathbf{X}, \mathbf{y})$ in \eqref{f_diff}, we have
\begin{equation}
    \nabla_{\mathbf{X}}f(\mathbf{X}, \mathbf{y}) = \nabla_{\mathbf{X}}\widetilde{f}(\mathbf{X}) \circ \mathbf{Y},
\end{equation}
where $\mathbf{Y} = \mathbf{1}_{K}\mathbf{y}^{T}$ and $\nabla_{\mathbf{X}}\widetilde{f}(\mathbf{X})$ is the derivative of $\widetilde{f}(\mathbf{X}) \triangleq \sum_{m=1}^{M}f_{m}(\mathbf{X})$, determined by
\begin{equation}
    \dfrac{\partial \widetilde{f}}{\partial x_{km}} = -\sum_{b_{l}\in \mathcal{B}_{m}}\beta_{k}^{m}\theta_{lk}\alpha_{l}^{m}.
\end{equation}
Following Lemma \ref{Hadamard_Cauchy} and some simple calculations, noting that $K \geq M$, we can obtain $\kappa_{11}$ and $\kappa_{12}$ in \eqref{lp} as
\begin{subequations}
\begin{align}
    \label{kappa11_kappa12}
    &\kappa_{11} = \sqrt{M}, \\
    &\kappa_{12} = \dfrac{PL_{\mathrm{max}}\sqrt{M}}{N_{0}}\mathop{\mathrm{max}}\limits_{l,k}\{q_{lk}^{2}\},
\end{align}
\end{subequations}
where $L_{\mathrm{max}} = \max\{L_{m}| m=1, \cdots, M\}$.

\subsubsection{AGP based Load-Balanced Clustered Cell-Free Networking}
Building upon the AGP algorithm presented in Section~\ref{AGP} and the parameter configuration in Section \ref{PC}, by employing balanced $k$-means \cite{malinen2014balanced} to obtain the BS partition, an AGP based load-balanced clustered cell-free networking algorithm is proposed, which is outlined in Algorithm \ref{algorithm4}.
\begin{figure}[t]
    \renewcommand{\algorithmicrequire}{\textbf{Input:}}
    \renewcommand{\algorithmicensure}{\textbf{Output:}}
	\begin{algorithm}[H]
		\caption{AGP-based $\alpha$-Fair Clustered Cell-Free Networking}
		\label{algorithm4}
		\begin{algorithmic}[1]
			\REQUIRE Large-scale fading matrix $\mathbf{Q}=[q_{lk}]\in \mathbb{R}^{L\times K}$, $\lambda$ in \eqref{lambda}, $\rho = 1.1$, $\mu_{t}= 5K$, the number of users $K$, BSs $L$ and subnetworks $M$, and the iteration numbers $T$, $J$ and $S$.
			\ENSURE Network decomposition $\mathcal{M}$.
			\STATE{\textbf{Initialization:} $\mathbf{X}_{0}\in \varOmega$, $\mathbf{y}_{0} \in \Delta$}.
                \STATE{Group BSs to obtain decomposition $\mathcal{B} = \{\mathcal{B}_{1}, \mathcal{B}_{2}, ..., \mathcal{B}_{M}\}$.}
			\FOR{$t = 0, ..., T-1$, }
                \STATE{Compute $\alpha_{l}^{m}$ and $\beta_{k}^{m}$ by $J$ alternating iterations with \eqref{alpha_beta}.} \label{algorithm_compute_alpha_beta}
                \STATE{Compute $f_{m}(\mathbf{X}_{t})$ by \eqref{sum_capacity_deterministic_equivalents}.}
                 \STATE{Compute $\gamma_{t}$ by \eqref{gamma}.}
			\STATE{Compute $\mathbf{y}_{t+1} = \mathrm{arg}\max_{\mathbf{y}\in \Delta} f_{\lambda, \gamma_{t}}(\mathbf{X}_{t}, \mathbf{y})$ by Lemma \ref{find_y_t}.} \label{algorithm_compute_y_t1}
                \STATE{Compute $\nabla_{\mathbf{X}}f(\mathbf{X}_{t}, \mathbf{y}_{t+1})$ by \eqref{f_diff} and $\zeta_{t}$ by \eqref{zeta}.}
			\STATE{Compute $\mathbf{X}^{(0)}_{t+1} = \mathbf{X}_{t} - \zeta_{t}\nabla_{\mathbf{X}}f(\mathbf{X}_{t}, \mathbf{y}_{t+1})$.}
                \FOR{$s = 0, ..., S$, }
                \STATE{Compute $\mathbf{X}^{(s)}_{t+1}$ and $\mathbf{Z}^{(s)}$ by \eqref{Dykstra} and \eqref{Omega_all}.} \label{algorithm_compute_As_Bs}
                \ENDFOR
                \STATE{Update $\mathbf{X}_{t+1} = \mathbf{X}^{(S)}_{t+1}$.}
			\ENDFOR
                \STATE Let $\mathbf{X}^{*} = \mathbf{X}_{T}$.
                \FOR{$k = 1, ..., K$, }
			\STATE{Let $m^{*} = \mathrm{arg}\max_{1\leq m \leq M}x^{*}_{km}$ and assign $u_{k}$ to $\mathcal{S}_{m^{*}}$.}
                \ENDFOR
                \RETURN{Network decomposition $\mathcal{M}= \{\mathcal{S}_{1}, \mathcal{S}_{2}, ..., \mathcal{S}_{M}\}$}.
		\end{algorithmic}
	\end{algorithm}
 \vspace{-6mm}
\end{figure}
The convergence of the proposed algorithm is demonstrated in the following Theorem \ref{AGP_convergence}.
\begin{thm}\label{AGP_convergence}
    For a given scalar $\varepsilon > 0$ and the parameters $\lambda$, $\gamma_{t}$ and $\zeta_{t}$ determined by \eqref{parameter_configuration}, there exists $t \leq T$, where $T \geq \mathcal{O}(\varepsilon^{-3})$, such that $(\mathbf{X}_{t}, \mathbf{y}_{t+1})$ is an $\varepsilon$-FNE of $\mathcal{P}6$.
\end{thm}
\begin{IEEEproof}
    For a rigorous proof of this theorem, the readers are referred to Theorem 1 in \cite{pan2021efficient}. However, it should be noted that the convergence of the AGP algorithm in \cite{pan2021efficient} is proven in the case with Euclidean projection by showing that 
    \begin{equation}
        \label{variational_characterization}
        z \in P_{\mathcal{S}}(x) \Leftrightarrow \langle \nabla f(x) - \nabla f(z), y - z \rangle \leq 0,\ \forall y\in \mathcal{S},
    \end{equation}
    where $\mathcal{S}$ is a convex set and $P_{\mathcal{S}}(x)$ denotes the projection of $x$ onto $\mathcal{S}$. By contrast, we utilize KL projection instead of Euclidean projection  in our Algorithm \ref{algorithm4}. Since KL projection and Euclidean projection are two special cases of Bregman projection, and the variational characterization given in \eqref{variational_characterization} is universally applicable to any Bregman projection \cite{hieu2022two}, \eqref{variational_characterization} also holds for KL projection. Therefore, our AGP based clustered cell-free networking algorithm converges. 
\end{IEEEproof}

Regarding the computational complexity of our algorithm,  balanced $k$-means is adopted to obtain the BS partition, which has the computational complexity of $\mathcal{O}(L^{3})$ \cite{malinen2014balanced}. Then for each iteration, we need $\mathcal{O}(JKL)$ complexity to iteratively compute all $\alpha_{l}^{m}$ and $\beta_{k}^{m}$ with $J$ iterations in line \ref{algorithm_compute_alpha_beta},  $\mathcal{O}(M)$  complexity to compute $\mathbf{y}_{t+1}$ in line \ref{algorithm_compute_y_t1} and $\mathcal{O}(SKM)$  complexity to iteratively compute $\mathbf{X}^{(s)}_{t+1}$ and $\mathbf{Z}^{(s)}$ with $S$ iterations in line \ref{algorithm_compute_As_Bs}. Therefore, the total computational cost of $T$ iterations is $\mathcal{O}(T(JKL + M + SKM))$ = $\mathcal{O}(KL)$.



\section{Simulation Results}\label{sec:Simulation Results}
In this section, simulation results are presented to demonstrate the performance of the proposed algorithms for the $\alpha$-fairness based clustered cell-free networking architecture. The small-scale fading is assumed to follow the standard complex Gaussian distribution, i.e., $g_{lk} \sim \mathcal{CN}(0,1)$. The large-scale fading is modeled as $q_{lk} = d_{lk}^{-\frac{\alpha_{0}}{2}}$, where $d_{lk}$ is the Euclidean distance between BS $b_{l}$ and user $u_{k}$, and $\alpha_{0}$ is the path-loss factor. The coordinates of the BSs and users are obtained through uniformly distributed random sampling within a unit square area. To eliminate the impact of the layout of BSs and users, 200 random layouts are generated to obtain the average results. The settings of the relevant parameters are summarized in Table \ref{Simulation Settings}.
\begin{table}[t]
    \caption{Simulation Settings}
    \vspace{-0.3cm}
    \begin{center}
	\begin{tabular}{c||c||c}
		\hline \bfseries Definition & \bfseries Symbol & \bfseries Value \\
		\hline
		Signal-to-noise ratio, SNR & $P/N_{0}$ & 0 dB\\
		Path loss factor & $\alpha_{0}$ & 4\\
		The regularization parameter & $\delta$ & 0.01 \\ 
            The iteration number of FFW & $T_{F}$ & 10 \\
            The iteration number of OPT  & $S_{F}$ & 10 \\
            The iteration number of deterministic equivalents & $N_{F}$ & 2 \\
            Step size of CCRP & $S_{C}$ & 0.1 \\
            The iteration number of CCRP & $T_{C}$ & 10 \\
            Step size of GNCCP & $S_{G}$ & 0.1 \\
            The iteration number of GNCCP & $T_{G}$ & 20 \\
            The iteration number of AGP & $T_{A}$ & 10000 \\
            The iteration number of Dykstra’s algorithm  & $S_{A}$ & 5 \\
		\hline
	\end{tabular}
	\label{Simulation Settings}
    \end{center}
\end{table}

We further compare our proposed $\alpha$-fairness approach with two representative baselines from the literature. The model in \cite{tight} corresponds to the special case of $\alpha=0$ in our formulation, i.e., maximizing the total sum ergodic capacity, but it is solved using the Bregman proximal gradient (BPG) algorithm with Dykstra's projection. Similarly, the model in \cite{9839089} corresponds to the limiting case of $\alpha \to \infty$, i.e., maximizing the minimum subnetwork sum ergodic capacity, yet it relies on the heuristic capacity-guaranteed networking (CGN) algorithm.

To comprehensively evaluate the performance of the proposed algorithms, we adopt three key metrics that jointly capture the fundamental efficiency-fairness trade-off in clustered cell-free networking. These metrics are defined as follows:
\begin{itemize}
    \item The ``Jain’s Index" or called ``Fairness Index" (FI) \cite{jain1984quantitative} is used to evaluate the fairness of the system, with a larger FI indicating a higher level of fairness, which is determined by
    \begin{equation}
        \mathrm{FI}(\{C_{m}\}_{m=1}^{M}) \triangleq \dfrac{ \left( \sum_{m=1}^{M} C_{m}  \right)^{2} }{ M \sum_{m=1}^{M}C_{m}^{2} } \in \left[\dfrac{1}{M}, 1 \right].
    \end{equation}

    \item Network sum ergodic capacity,
    \begin{equation}
        C_{\mathrm{sum}}(\{C_{m}\}_{m=1}^{M}) = \sum_{m=1}^{M}C_{m},
    \end{equation}
    is used to measure the overall performance of the system.

    \item Minimum subnetwork sum ergodic capacity,
    \begin{equation}
        C_{\mathrm{min}}(\{C_{m}\}_{m=1}^{M}) = \min_{m\in \{1,2,\cdots, M\}} C_{m},
    \end{equation}
    is employed to characterize the system’s minimum achievable performance, thereby reflecting the worst-case subnetwork state.
\end{itemize}

\subsection{Simulation Results Across Different Cases}\label{sec:Simulation Results}

\begin{figure}[t]
    \centering
    \hspace{-1em}
	\subfigure[]{\label{250808_case1_FI}
		\includegraphics[width=0.32\linewidth]{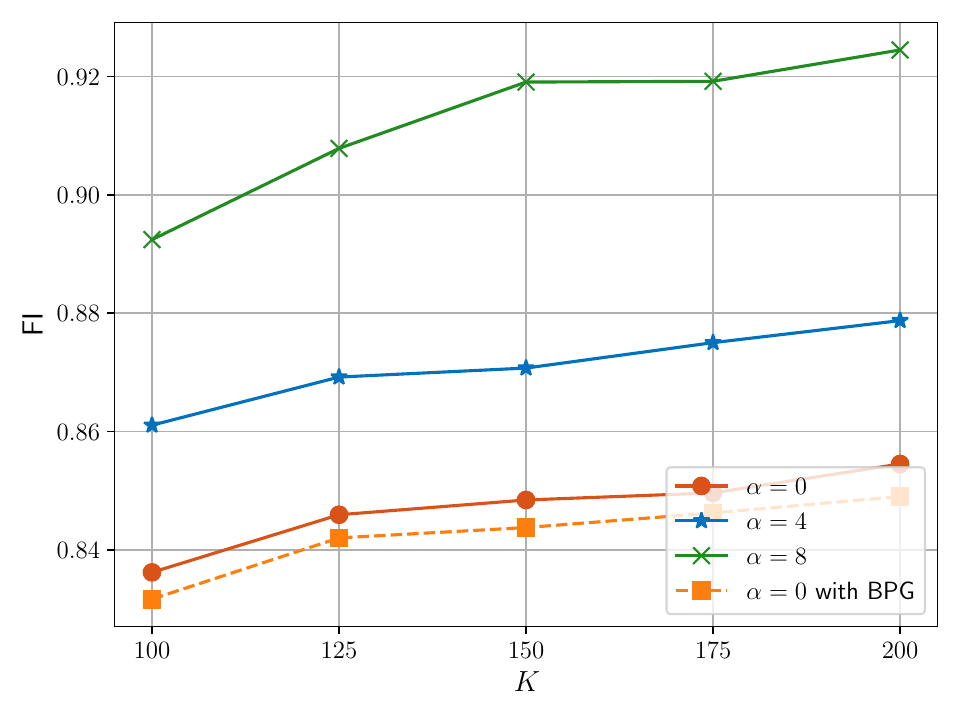}}
	\subfigure[]{\label{250808_case1_Cmin}
		\includegraphics[width=0.32\linewidth]{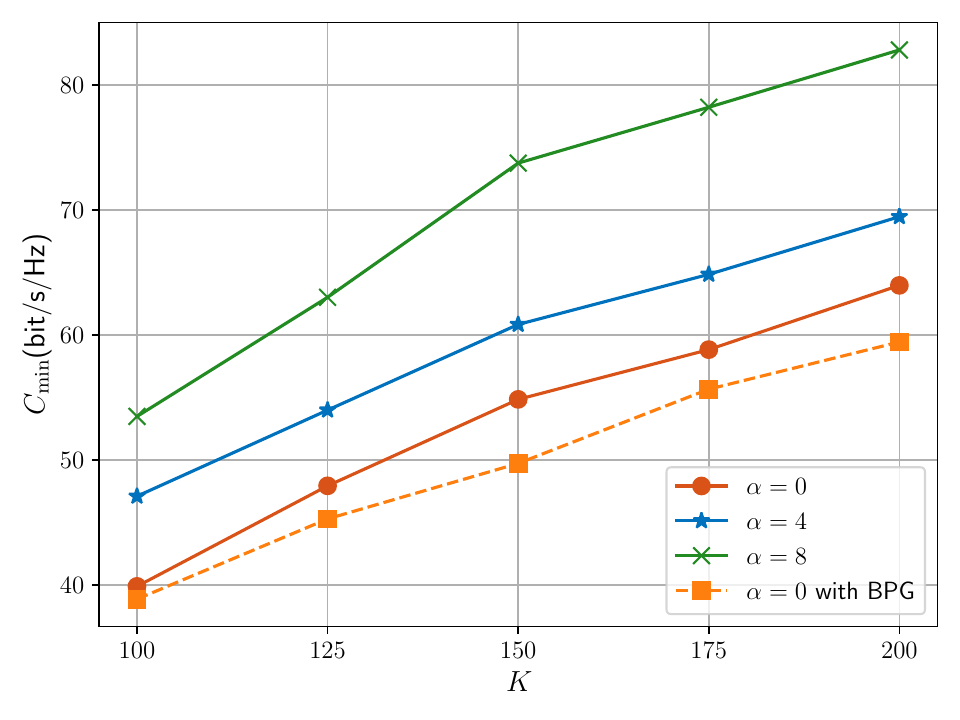}}
        \subfigure[]{\label{250808_case1_Csum}
		\includegraphics[width=0.32\linewidth]{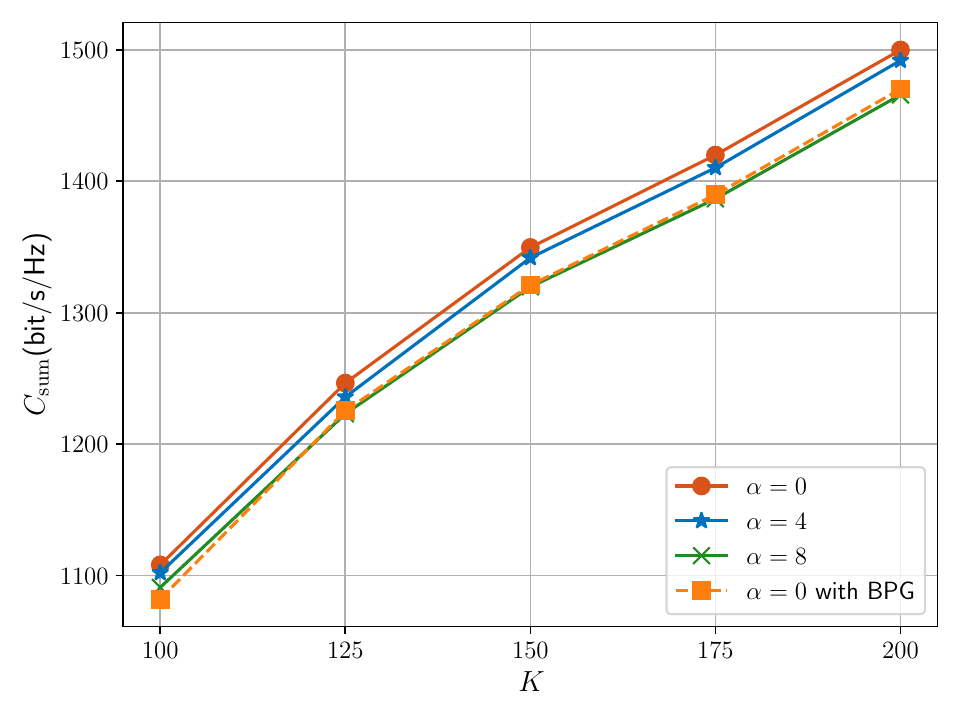}}
	\caption{In $Case$ 1, Jain’s index FI, minimum subnetwork sum ergodic capacity $C_{\mathrm{min}}$ and total subnetwork sum ergodic capacity $C_{\mathrm{sum}}$ with $\alpha = 0, 4, 8$ and $L = 150$, $M = 10$.}
    \label{fig_case1}
\end{figure}

\begin{figure}[t]
    \centering
    \hspace{-1em}
	\subfigure[]{\label{250808_case2_FI}
		\includegraphics[width=0.32\linewidth]{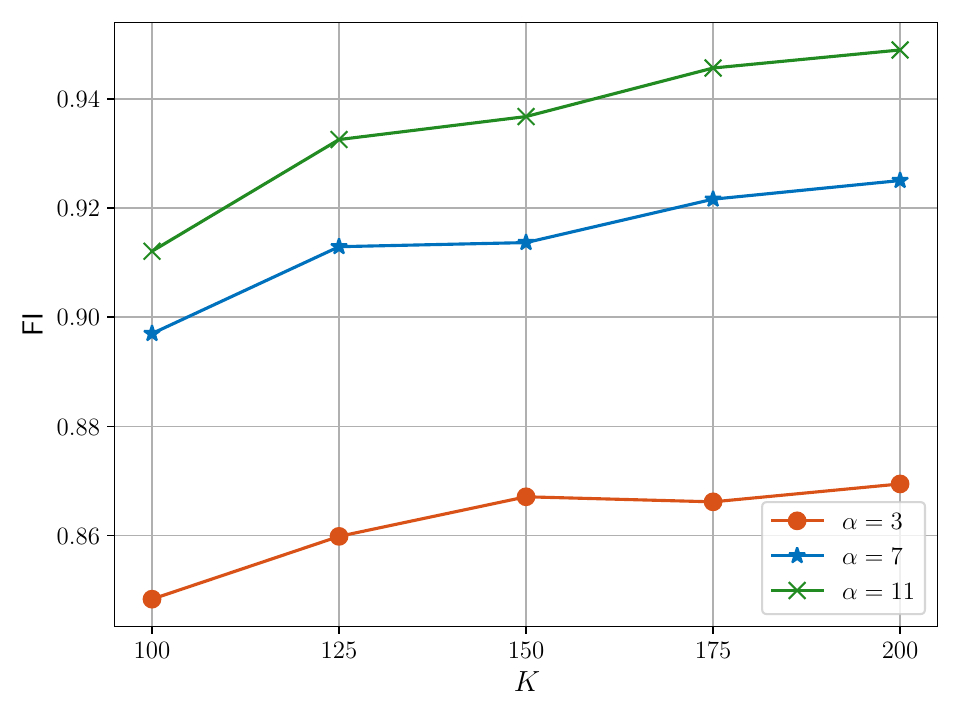}}
	\subfigure[]{\label{250808_case2_Cmin}
		\includegraphics[width=0.32\linewidth]{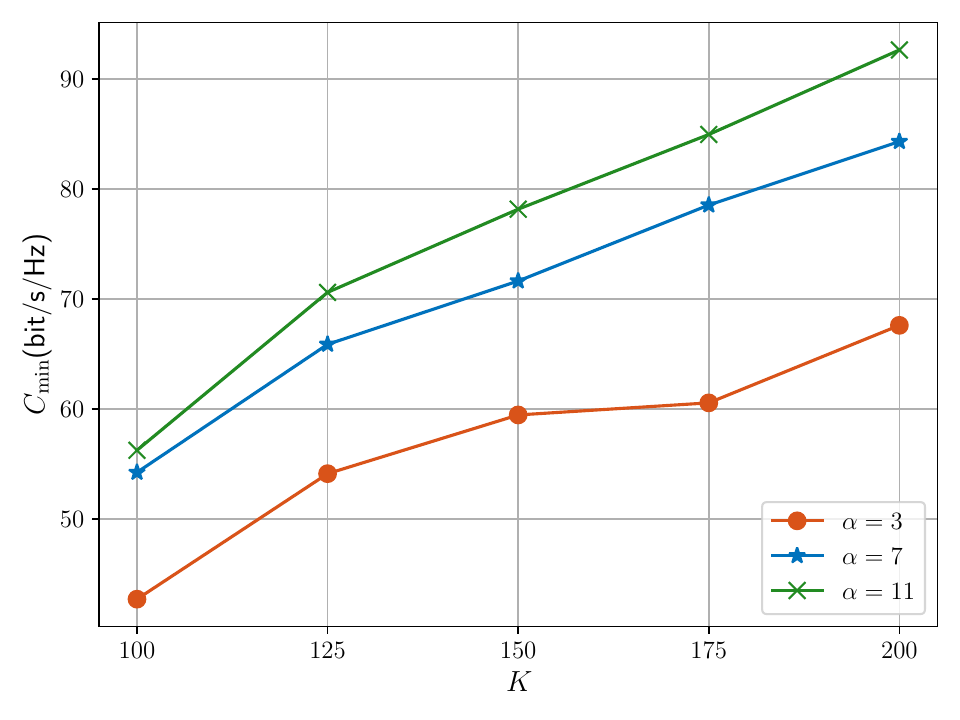}}
        \subfigure[]{\label{250808_case2_Csum}
		\includegraphics[width=0.32\linewidth]{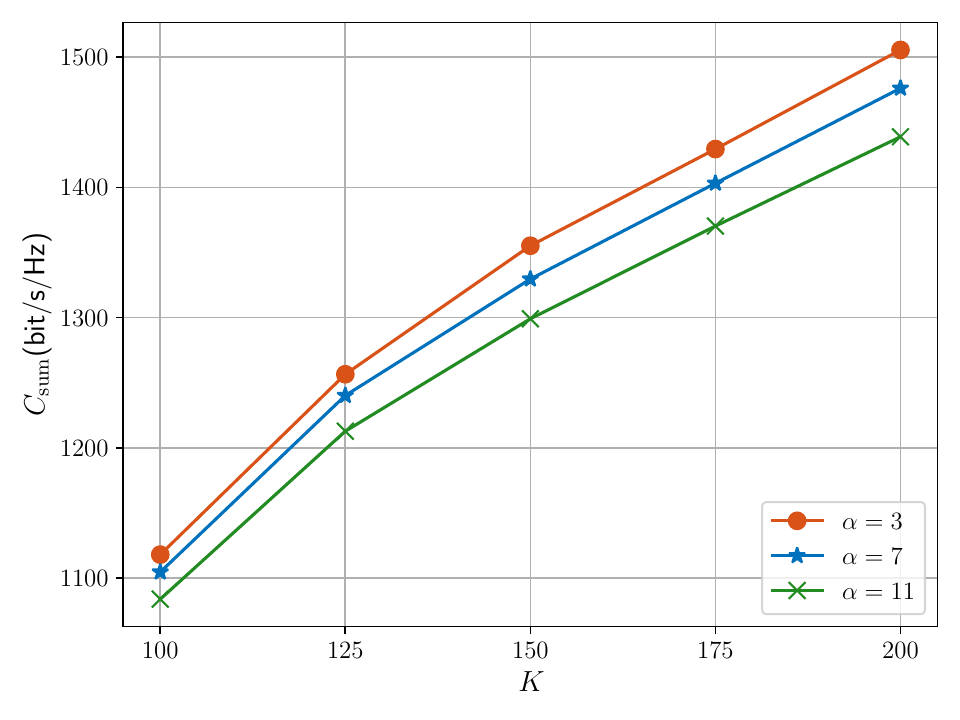}}
	\caption{In $Case$ 2, Jain’s index FI, minimum subnetwork sum ergodic capacity $C_{\mathrm{min}}$ and total subnetwork sum ergodic capacity $C_{\mathrm{sum}}$ with $\alpha = 3, 7, 11$ and $L = 150$, $M = 10$.}
    \label{fig_case2}
\end{figure}

\begin{figure}[t]
    \centering
    \hspace{-1em}
	\subfigure[]{\label{250808_case3_FI}
		\includegraphics[width=0.32\linewidth]{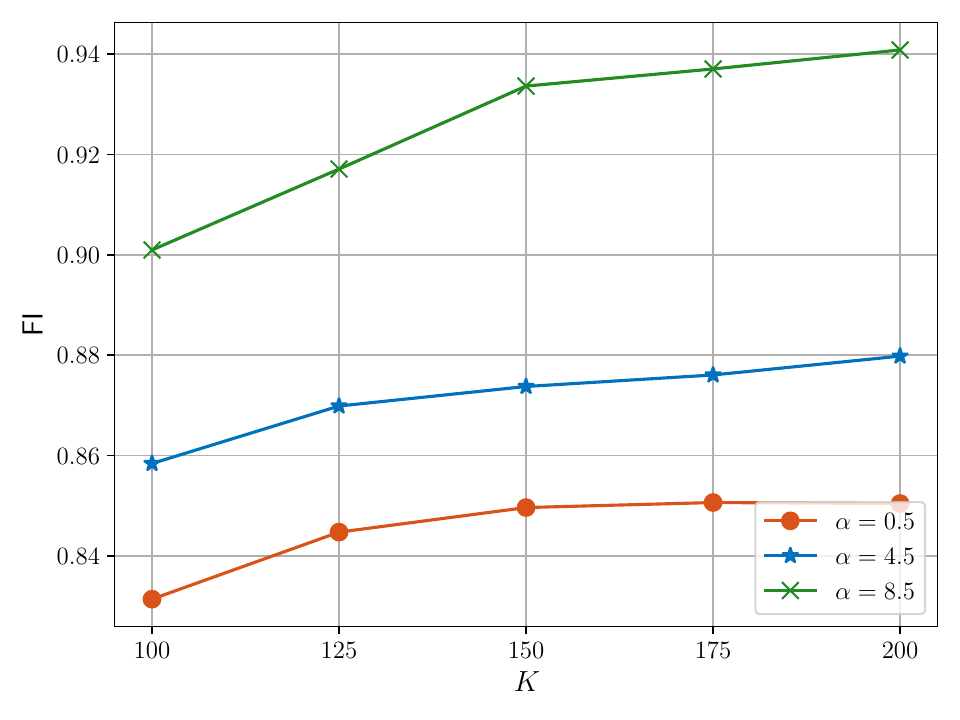}}
	\subfigure[]{\label{250808_case3_Cmin}
		\includegraphics[width=0.32\linewidth]{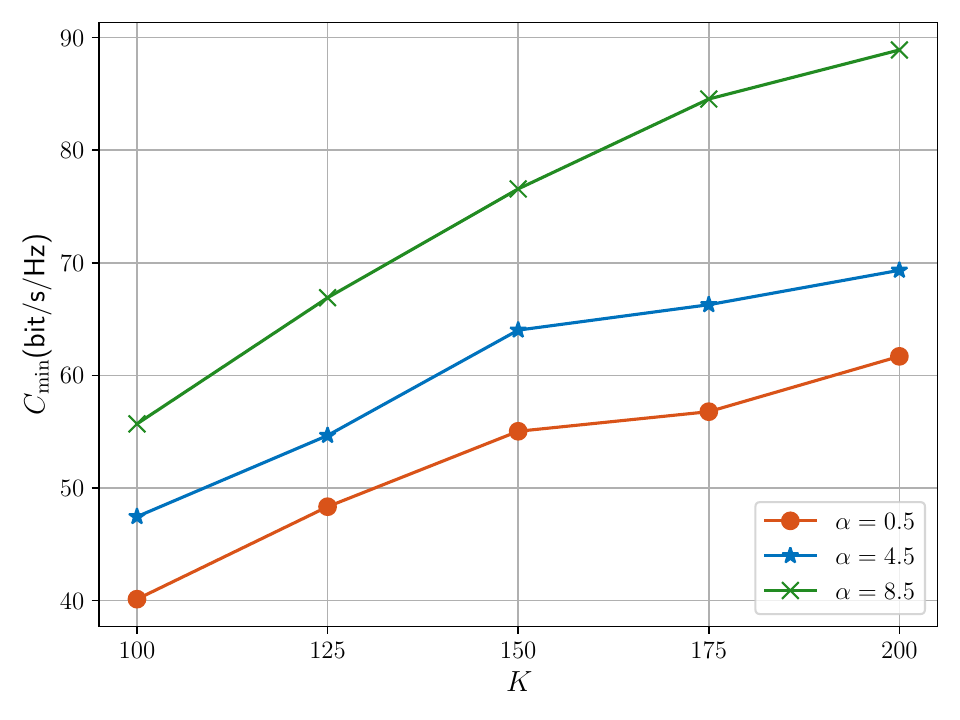}}
        \subfigure[]{\label{250808_case3_Csum}
		\includegraphics[width=0.32\linewidth]{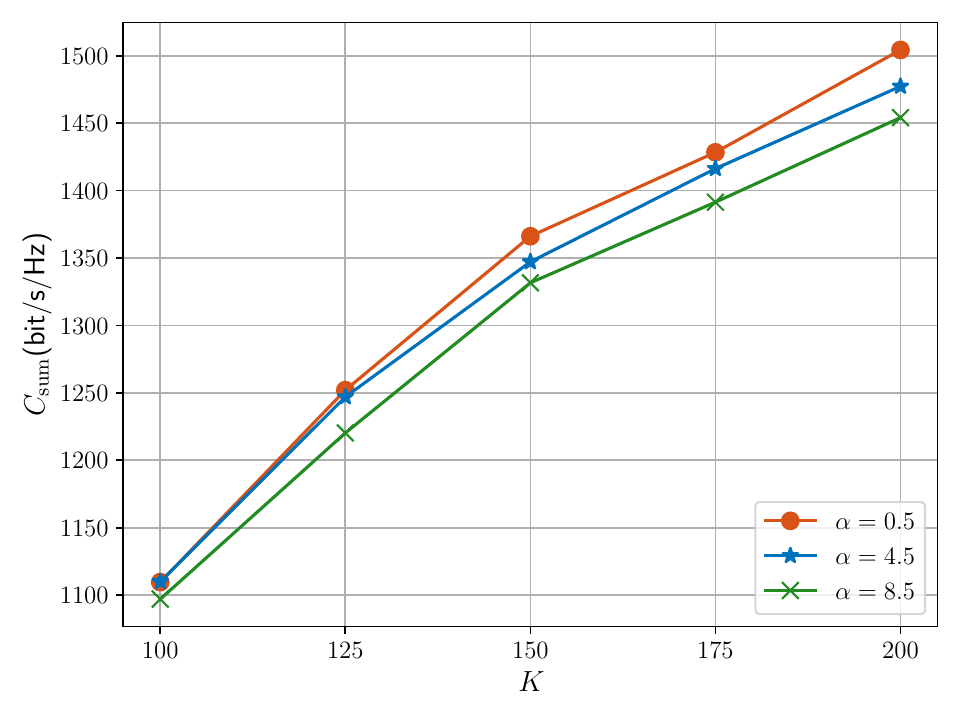}}
	\caption{In $Case$ 3, Jain’s index FI, minimum subnetwork sum ergodic capacity $C_{\mathrm{min}}$ and total subnetwork sum ergodic capacity $C_{\mathrm{sum}}$ with $\alpha = 0.5, 4.5, 8.5$ and $L = 150$, $M = 10$.}
    \label{fig_case3}
\end{figure}

\begin{figure}[t]
    \centering
    \hspace{-1em}
	\subfigure[]{\label{250808_case4_FI}
		\includegraphics[width=0.32\linewidth]{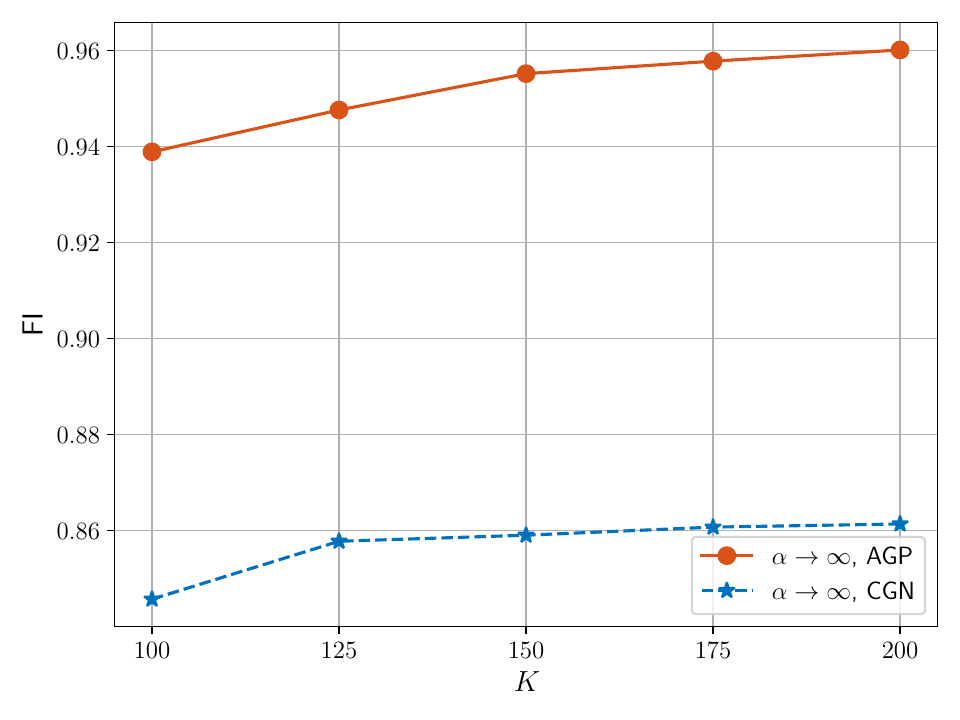}}
	\subfigure[]{\label{250808_case4_Cmin}
		\includegraphics[width=0.32\linewidth]{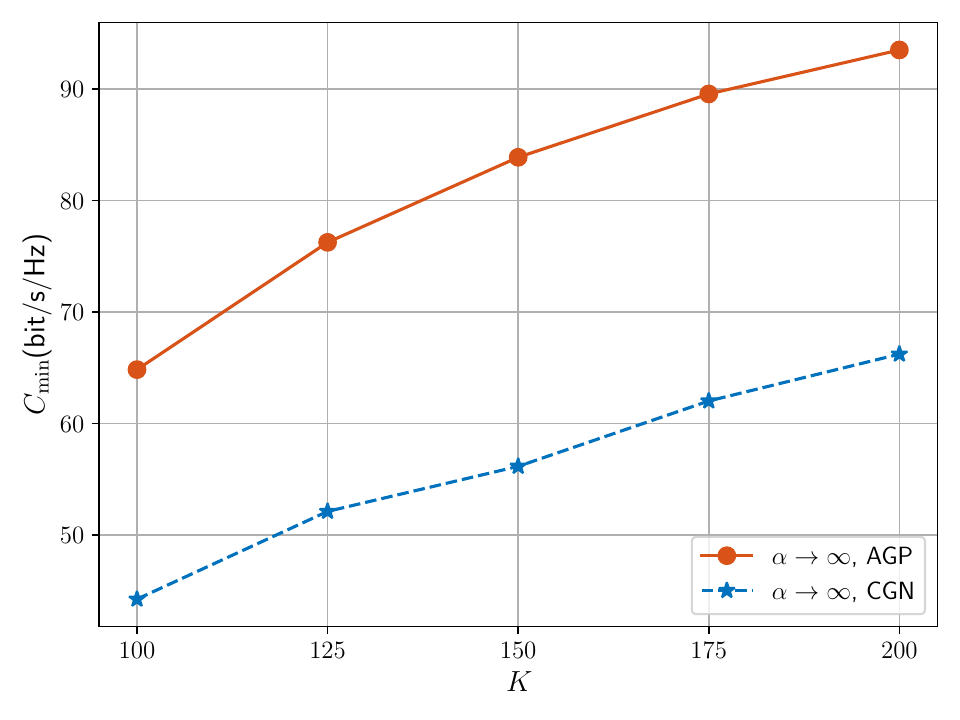}}
        \subfigure[]{\label{250808_case4_Csum}
		\includegraphics[width=0.32\linewidth]{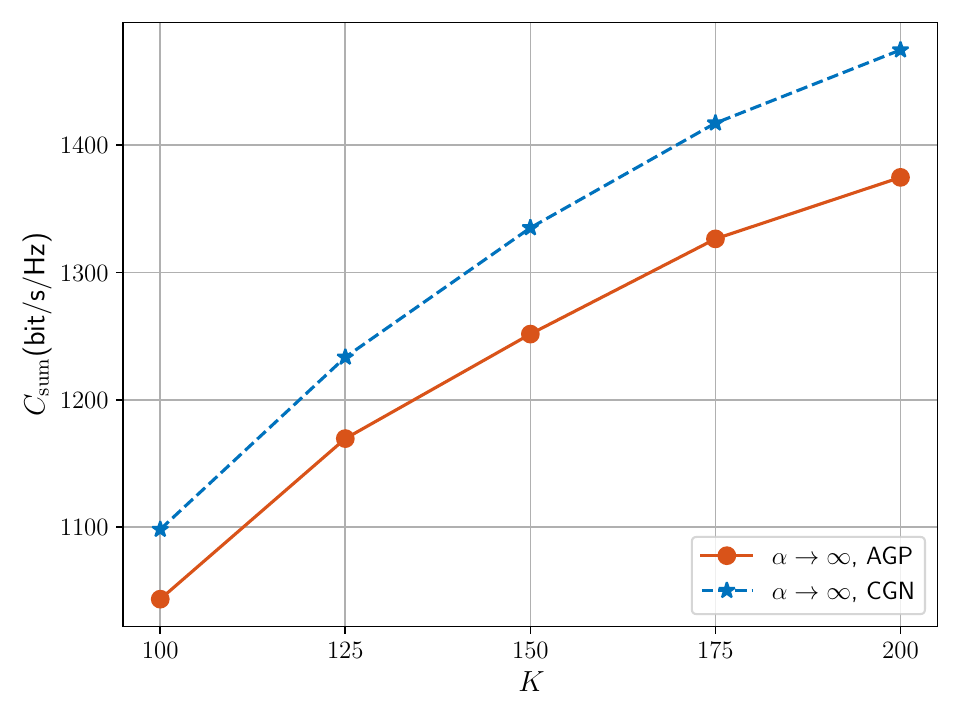}}
	\caption{In $Case$ 4, Jain’s index FI, minimum subnetwork sum ergodic capacity $C_{\mathrm{min}}$ and total subnetwork sum ergodic capacity $C_{\mathrm{sum}}$ with proposed AGP and CGN \cite{9839089} algorithms. $L = 150$, $M = 10$.}
    \label{fig_case4}
\end{figure}

For different cases with varying values of $\alpha$, we propose corresponding approaches. In this subsection, we experimentally evaluate the feasibility of the four proposed approaches. Specifically, for each case, different values of $\alpha$ are selected to examine the performance of the three metrics under various scenario configurations. We choose $\alpha = 0, 4, 8$ for $Case$ 1, where $\alpha = 0$ uses the algorithm described in \cite{tight} as the baseline, $\alpha = 3, 7, 11$ for $Case$ 2, and $\alpha = 0.5, 4.5, 8.5$ for $Case$ 3. As for $Case$ 4, it solely involves the scenario where $\alpha \rightarrow \infty$, and we use the CGN algorithm in \cite{9839089} as the baseline for comparison with our proposed AGP algorithm discussed in Section \ref{AGP}. After 200 random layouts, we obtained the FI, $C_{\mathrm{min}}$, and $C_{\mathrm{sum}}$ values for different values of $\alpha$ with the four case under different numbers of users $K$, as shown in Fig. \ref{fig_case1} - Fig. \ref{fig_case4}. 

In $Case$ 1, the results demonstrate a clear trade-off between fairness and overall performance, with different values of $\alpha$ plotted to show this relationship. Specifically, Fig. \ref{250808_case1_FI} shows that larger $\alpha$ values correspond to larger FI values, indicating improved fairness. This is corroborated by Fig. \ref{250808_case1_Cmin}, which illustrates that higher fairness implies a larger $C_{\mathrm{min}}$ as the system places greater emphasis on users in the worst conditions. The trade-off, shown in Fig. \ref{250808_case1_Csum}, is a slight reduction in the overall system quality. For this analysis, we include $\alpha = 4$ and $\alpha = 8$ to demonstrate generality. The $\alpha = 0$ case is also plotted as a baseline; this is notable because it reduces the original problem $\mathcal{P}2$ to the max-sum problem $\max \sum_{m=1}^{M}f_{m}(\mathbf{X})$, which was studied in \cite{tight} using the BPG with Dykstra’s algorithm. Quantitatively, as $\alpha$ increases from 0 to 8, FI and $C_{\mathrm{min}}$ increase by approximately 8\% and 32\%, respectively, while the overall quality decreases by only 2\%. This highlights that a substantial improvement in fairness can be achieved with only a modest degradation in overall system performance.


As for $Case$ 2 and $Case$ 3, we assess the feasibility of the proposed method by evaluating its performance across different $\alpha$ values, demonstrating that it consistently delivers superior results. Specifically, in $Case$ 2 as shown in Fig. \ref{fig_case2}, FI and $C_{\mathrm{min}}$ increase with $\alpha$, accompanied by a slight decrease in $C_{\mathrm{sum}}$. When $\alpha$ increases from 3 to 7, FI rises by 8\%, $C_{\mathrm{min}}$ by 34\%, while $C_{\mathrm{sum}}$ decreases by only 3\%. Similarly, in $Case$ 3 shown in Fig. \ref{fig_case3}, increasing $\alpha$ from 0.5 to 8.5 yields a 10\% increase in FI, a 42\% increase in $C_{\mathrm{min}}$, and only a 2\% reduction in $C_{\mathrm{sum}}$.

As for $Case$ 4, when $\alpha \rightarrow \infty$, the original problem $\mathcal{P}2$ reduces to the max–min problem, which has been discussed in \cite{9839089} with a heuristic algorithm. We use this heuristic as the baseline and compare it with the equivalent relaxed algorithm proposed in this paper, as shown in Fig. \ref{fig_case4}. Compared with the CGN algorithm, the proposed AGP algorithm achieves an 11\% improvement in FI and a 45\% increase in $C_{\mathrm{min}}$, with only a 5\% reduction in $C_{\mathrm{sum}}$. This demonstrates that our algorithm not only significantly enhances system fairness and the performance of disadvantaged users but also achieves a balanced performance with minimal loss in overall system quality.

\begin{figure*}[t]
    \centering
	\subfigure[]{\label{250808_alpha_FI}
		\includegraphics[width=0.32\linewidth]{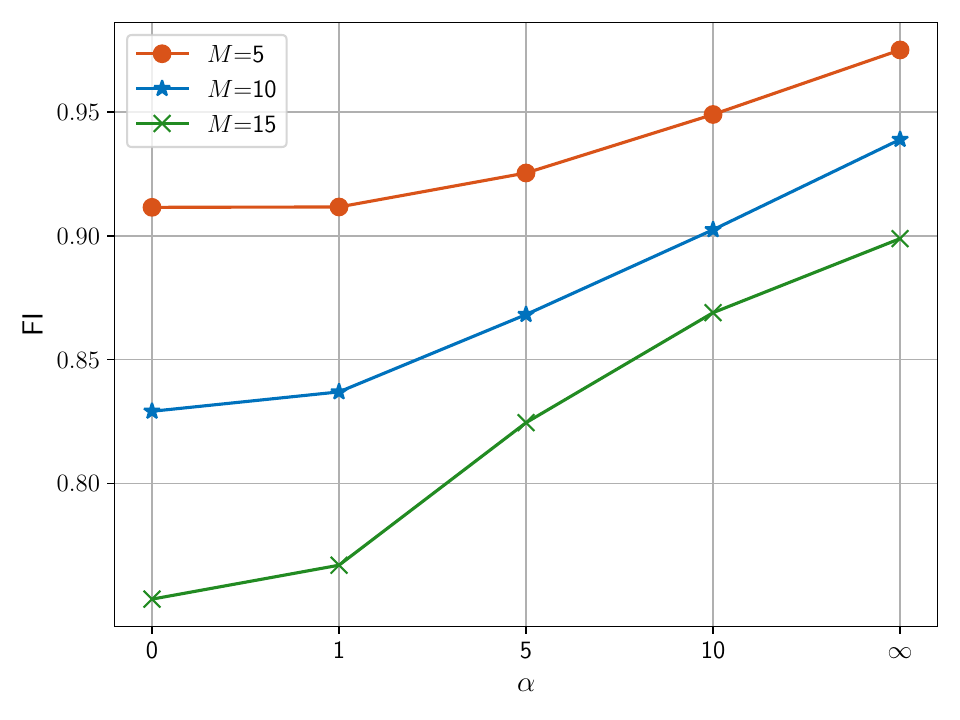}}
	\subfigure[]{\label{250808_alpha_Cmin}
		\includegraphics[width=0.32\linewidth]{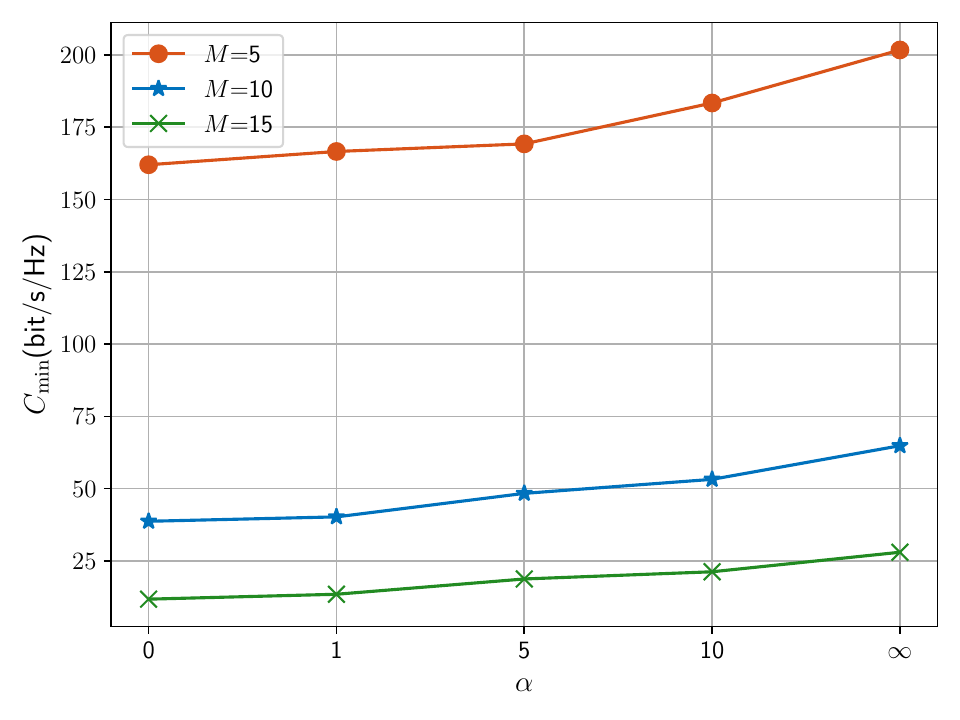}}
        \subfigure[]{\label{250808_alpha_Cavg}
		\includegraphics[width=0.32\linewidth]{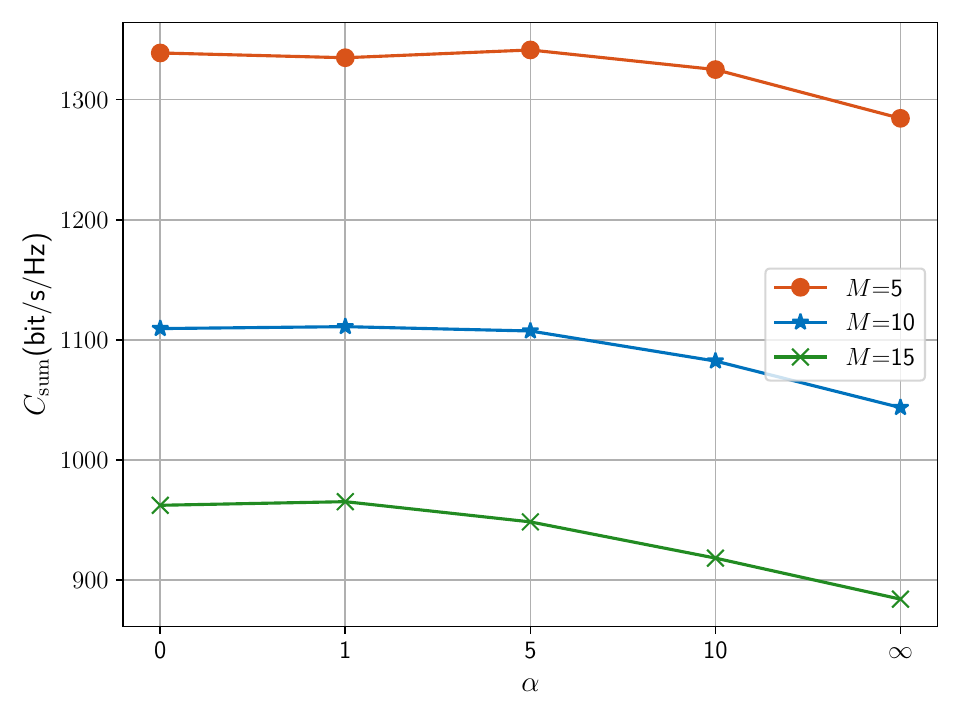}}
	\caption{Jain’s index FI, total subnetwork sum ergodic capacity $C_{\mathrm{sum}}$ and minimum subnetwork sum ergodic capacity $C_{\mathrm{min}}$ in different cases of $\alpha$ with $K = 100$, $L = 150$ and $M = 5, 10, 15$.}
    \label{fig_alpha}
\end{figure*}

\subsection{Effect of the Fairness Parameter \(\alpha\)}
As shown in Fig. \ref{fig_alpha}, a representative value of $\alpha$ is selected for each case, i.e., $\alpha = 1$, $5$, and $10$ for $Cases$ 1, 2, and 3, respectively, and $\alpha \rightarrow \infty$ for $Case$ 4. Similarly, $\alpha = 0$ is used as the baseline.

As mentioned earlier, Jain's index measures the fairness of the system, with a higher FI indicating a greater level of fairness. This trend is consistent with the values of $\alpha$ in different cases. Fig. \ref{250808_alpha_FI} further illustrates this behavior with different number of subnetworks $M$. For example, when $M = 15$, as $\alpha$ increases from 0 to $\infty$, FI rises by approximately 19\%, and $C_{\mathrm{min}}$ increases from 11 bit/s/Hz to 28 bit/s/Hz, representing a growth of over 137\%, while the total system sum capacity $C_{\mathrm{sum}}$ decreases by only 8\%.

It can be observed from Fig.~\ref{250808_alpha_FI} that, when the number of users $K$ and BSs $L$ are fixed, the Jain’s index FI decreases as the number of subnetworks $M$ increases. This is because a larger $M$ leads to greater diversity in the sum capacities across subnetworks. While our optimization procedure yields certain improvements in fairness, for large $M$ the ratio between the maximum sum capacity and the minimum sum capacity $C_{\min}$ remains relatively high. As noted in \cite{jain1984quantitative}, system fairness decreases as this ratio increases. Consequently, smaller $M$ values yield higher fairness, as subnetwork capacities tend to be more balanced.

Similarly, as shown in Section \ref{sec:Simulation Results}, for a fixed number of BSs $L$ and a fixed number of subnetworks $M$, the Jain’s index FI increases with the number of users $K$. This trend arises because adding more users provides the system with more opportunities to enhance $C_{\min}$, as evident from the $C_{\min}$–$K$ curves in Fig.~\ref{fig_case1}–Fig.~\ref{fig_case1}. In other words, with a larger user pool, the scheduler can better balance resource allocation among subnetworks, thereby improving fairness. However, this improvement may eventually plateau as $K$ grows, since capacity gains for the weakest subnetworks become increasingly limited once they approach their performance ceiling.

\subsection{Visualization}

\begin{figure}[t]
    \centering
    
    \subfigure[$\alpha=1$]{\label{fig:alpha1}
        \includegraphics[width=0.465\linewidth]{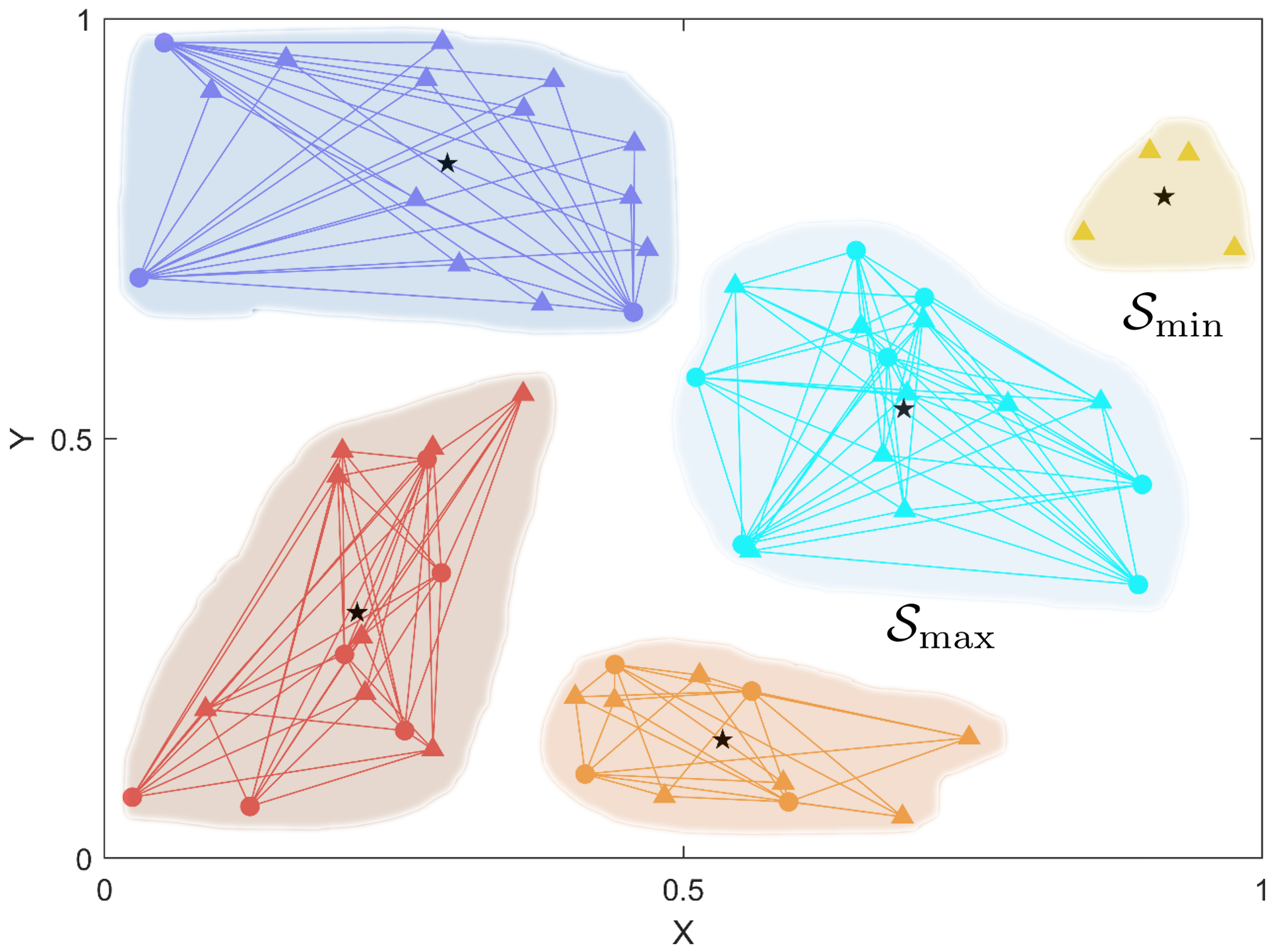}
    }
    \hfill
    \subfigure[$\alpha=5$]{\label{fig:alpha5}
        \includegraphics[width=0.465\linewidth]{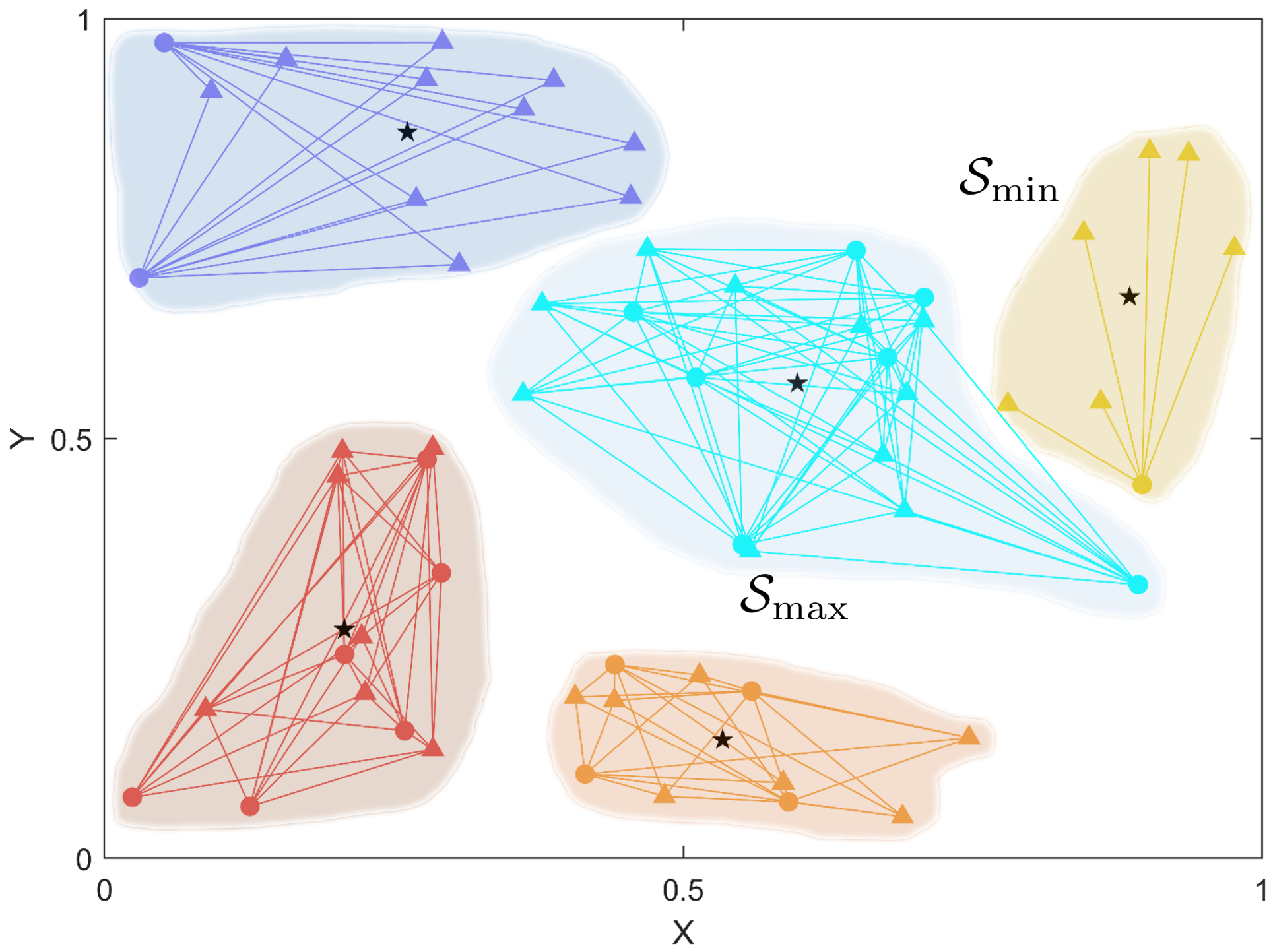}
    }
    \\[0.4cm]
    
    \subfigure[$\alpha=10$]{\label{fig:alpha10}
        \includegraphics[width=0.465\linewidth]{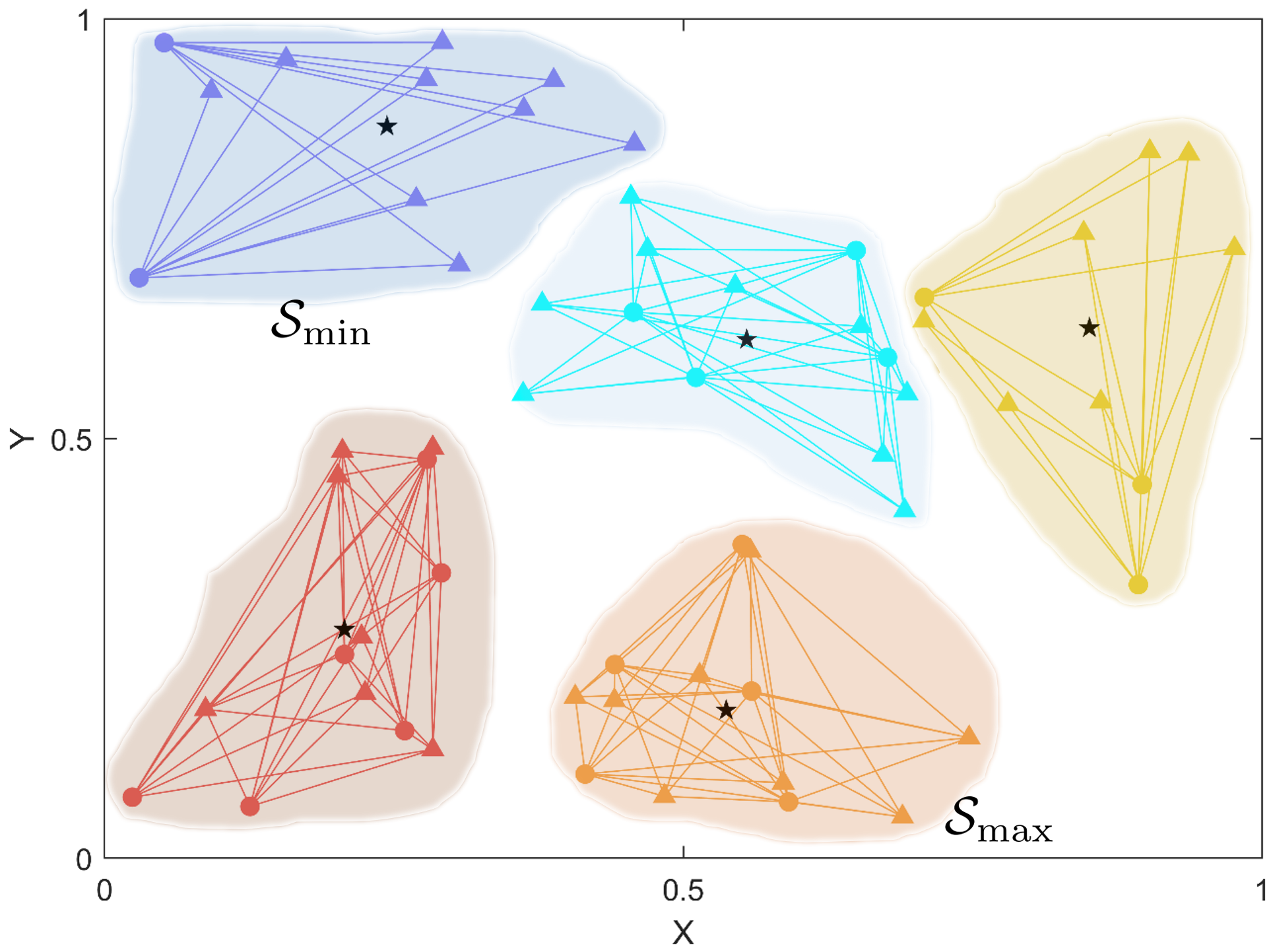}
    }
    \hfill
    \subfigure[$\alpha \to \infty$]{\label{fig:alpha-inf}
        \includegraphics[width=0.465\linewidth]{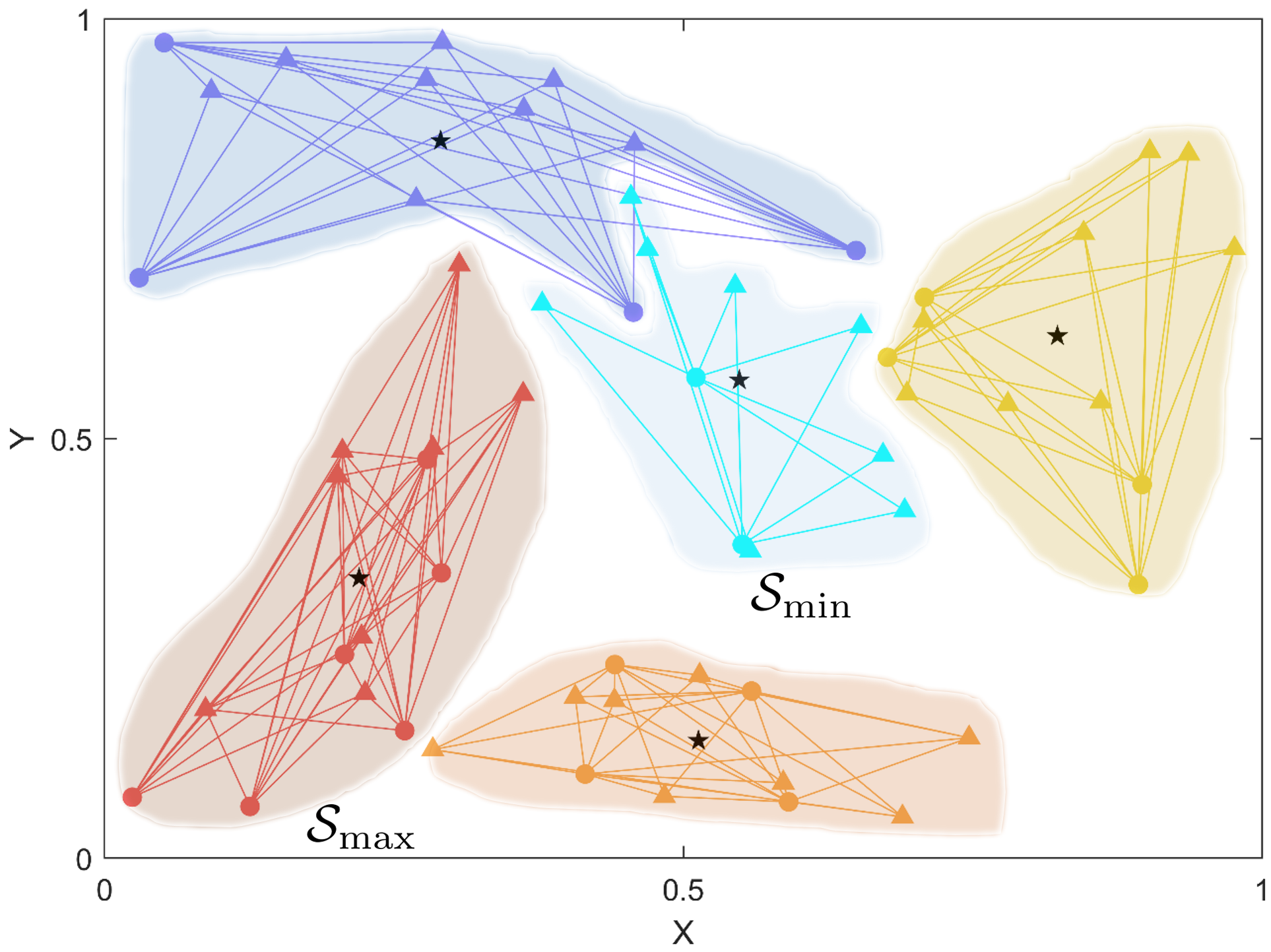}
    }
    
    \caption{Visualization of clustered cell-free networks under different fairness parameter $\alpha$ with $K=20$, $L=40$ and $M=5$. Both BSs and users are UD in the unit square. Triangles represent BSs, circles represent users, and black stars represent the centroids of subnetworks. 
    (a) $Case$ 3 with $\alpha=1$ (proportional fairness), 
    (b) $Case$ 2 with $\alpha=5$, 
    (c) $Case$ 1 with $\alpha=10$, 
    (d) $Case$ 4 with $\alpha \to \infty$ (max-min fairness).}
    \label{fig:subnetwork-partitioning}
\end{figure}

\begin{figure}[t]
    \centering
    \includegraphics[width=1\linewidth]{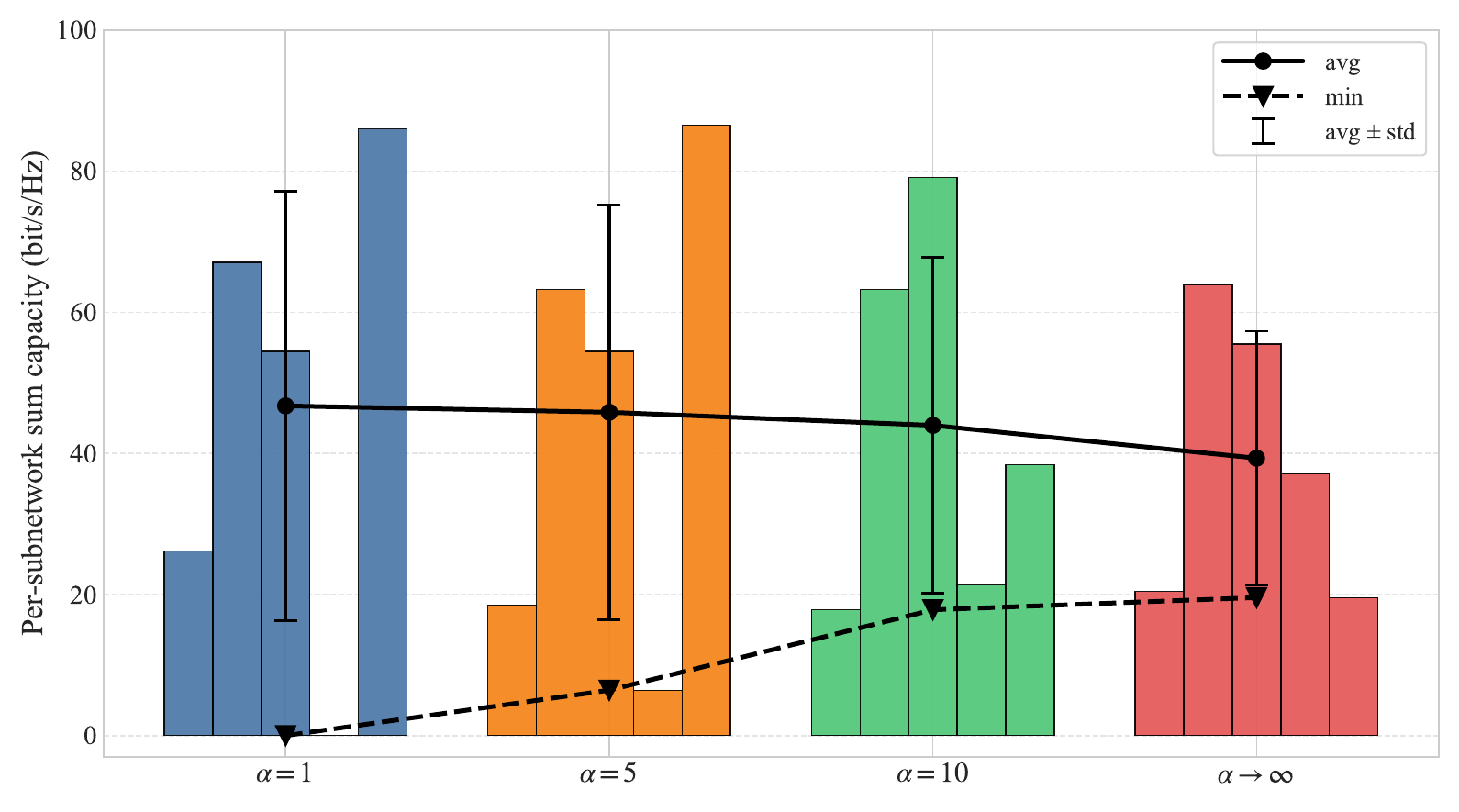}  
    \caption{Per-subnetwork sum capacity comparison of clustered cell-free networks under varying fairness parameter with $\alpha=1$ (proportional fairness), $\alpha=5$, $\alpha=10$, and $\alpha \to \infty$ (max-min fairness). Both BSs and users are UD in the unit square and $K=20$, $L=40$ and $M=5$}
    \label{fig:Per-subnetwork_sum_capacity}
\end{figure}

To intuitively illustrate the impact of the fairness parameter $\alpha$ on the network decomposition, we present visualization results and the corresponding per-subnetwork uplink ergodic sum capacity for a representative scenario with $K=20$ users, $L=40$ BSs, and $M=5$ subnetworks. Both BSs and users are uniformly distributed in the unit square.

Figure~\ref{fig:subnetwork-partitioning} displays the visualization of the network under different values of $\alpha$. Triangles represent BSs, circles denote users, and black stars indicate the centroids of the five subnetworks. As $\alpha$ increases from 1 (proportional fairness) to $\infty$ (max-min fairness), the partitioning exhibits a clear trend toward greater balance: the maximum subnetwork capacity of subnetwork $\mathcal{S}_{\max}$ gradually decreases, the minimum subnetwork capacity of subnetwork $\mathcal{S}_{\min}$ increases substantially, and the spread of the per-subnetwork capacity narrows markedly.

Of particular note is the case when $\alpha=1$ (proportional fairness), where the optimization objective tends to emphasize maximizing overall system performance with less emphasis on inter-subnetwork balance. In this scenario, one subnetwork in Fig. \ref{fig:alpha1} has a sum capacity of exactly 0, meaning that it contains only BSs with no users assigned. This results in a significant waste of resources, as a cluster of BSs remains idle. Such an extreme imbalance occurs because smaller $\alpha$ encourages concentrating users into fewer subnetworks to maximize aggregate capacity, at the expense of leaving some subnetworks severely underutilized.

As $\alpha$ progressively increases (e.g., from 5 to 10, and approaching $\infty$), the optimization increasingly penalizes subnetworks with extremely low capacity, forcing a more uniform user allocation across all subnetworks. Consequently, the sum capacity of the previously empty subnetwork steadily recovers, rising from 0 bit/s/Hz to 6.44 bit/s/Hz, then to 21.33 bit/s/Hz, and finally reaching 37.16 bit/s/Hz under max-min fairness, while the capacities of the other subnetworks also converge toward a narrower range. This evolution clearly demonstrates how the tunable $\alpha$-fairness scheme achieves a controllable trade-off between global system efficiency and inter-subnetwork load balancing: larger $\alpha$ values effectively suppress the occurrence of near-zero-capacity subnetworks and significantly improve overall network fairness.

Figure~\ref{fig:Per-subnetwork_sum_capacity} further quantifies this trend using a bar plot of the per-subnetwork sum ergodic capacities. As $\alpha$ increases, $\mathcal{S}_{\min}$ improves dramatically (from 0 to 19.56), $\mathcal{S}_{\max}$ decreases accordingly (from 85.97 to 63.94), and the capacity variance drops sharply (from 923.83 to 323.93). Meanwhile, the average capacity declines only modestly, underscoring the classic efficiency-fairness trade-off enabled by the adjustable $\alpha$-fairness scheme.

\section{Conclusion}\label{sec:Conclusion}
In this paper, we have investigated the fairness problem in clustered cell-free networking and proposed a unified and tunable \(\alpha\)-fairness scheme that effectively balances overall spectral efficiency and fairness across subnetworks. By approximating the ergodic sum capacity of each subnetwork with its closed-form deterministic equivalent, we transformed the original combinatorial clustering problem into a continuous optimization problem. Leveraging the concavity/convexity properties of the \(\alpha\)-fair objective, we rigorously classified the problem into four distinct cases according to the value of \(\alpha\). For each case, we established the exact equivalence between the original integer program and its continuous relaxation through tailored theoretical proofs, and developed efficient algorithms with guaranteed convergence (FFW, CCRP-FFW, GNCCP-FFW, and AGP).

Extensive simulations demonstrate the superiority of the proposed approach. Compared with state-of-the-art baselines, our \(\alpha\)-fairness scheme achieves up to 11\% improvement in Jain’s fairness index and 45\% gain in minimum subnetwork capacity, while incurring only a negligible 5\% reduction in aggregate throughput. These results confirm that the single tunable parameter \(\alpha\) enables network operators to smoothly navigate the entire efficiency-fairness trade-off spectrum in practical fronthaul-constrained clustered cell-free networking deployments.

The proposed scheme not only provides a theoretically sound and computationally efficient solution to fairness-aware clustered cell-free networking, but also offers a flexible and practical tool for future 6G wireless systems. Future work may extend the current model to incorporate dynamic user mobility, time-varying channels, and multi-cell coordination scenarios.

\appendices

\section{Proof of Theorem \ref{equivalent_P2_P6}}\label{equivalent_P2_P6_proof}
Before establishing the equivalence between $\mathcal{P}2$ and $\mathcal{P}6$, we need to introduce the following optimization problem
\begin{subequations}
    \begin{align}
\mathcal{P}7:
    \min_{\mathbf{X}} \max_{m} &\quad f_{m}(\mathbf{X})&& \label{problemP7a}\\
    \mbox{s.t.} 
    & \quad \mathbf{X} \in \varOmega_{01}, \label{problemP7b} 
    \end{align}
\end{subequations}
where the distinction between $\mathcal{P}6$ and $\mathcal{P}7$ lies in $\mathbf{X}$ being a discrete 0-1 variable rather than continuous.

We demonstrate the equivalence between $\mathcal{P}2$ and $\mathcal{P}6$ by establishing pairwise equivalence among $\mathcal{P}2$, $\mathcal{P}3$, $\mathcal{P}7$, and $\mathcal{P}6$, i.e.,

\begin{equation}  
\raisebox{1.6ex}{$
\left.  
     \begin{array}{lr}  
     \mathcal{P}3 \iff \mathcal{P}7 &  \\  
     \mathcal{P}6 \iff \mathcal{P}7 &    
     \end{array}  
\!\!\!\!\!\!\!\!\right\} 
\rightarrow\!\!
$}
\left.
  \begin{array}{c}  
  \mathcal{P}3 \iff \mathcal{P}6 \\ 
  \mathcal{P}2 \iff \mathcal{P}3  
  \end{array}
\!\!\right\}
\rightarrow
\mathcal{P}2 \iff \mathcal{P}6.
\end{equation}

First, the equivalence of $\mathcal{P}3$ and $\mathcal{P}7$ can be elucidated with the Lemma 3 in \cite{879343}.

As for $\mathcal{P}6$ and $\mathcal{P}7$, their equivalence can be established by the following Lemma \ref{equivalent_P6_P7}.
\begin{lemma}\label{equivalent_P6_P7}
    If $\mathcal{P}6$ has a Nash equilibrium, then $\mathcal{P}6$ and $\mathcal{P}7$ have the same optimal value and share a common optimal solution.
\end{lemma}
\begin{IEEEproof}
    Denoting the Nash equilibrium of $\mathcal{P}6$, i.e., the optimal solution to $\mathcal{P}6$, by $(\mathbf{X}^{*}, \mathbf{y}^{*})$,  we have
    \begin{equation}
        \label{Nash_equilibrium_P3}
        f(\mathbf{X}^{*}, \mathbf{y}) \leq f(\mathbf{X}^{*}, \mathbf{y}^{*}) \leq f(\mathbf{X}, \mathbf{y}^{*}), \forall \mathbf{X}\in \varOmega,\ \forall \mathbf{y}\in \Delta.
    \end{equation}
    According to the first inequality of \eqref{Nash_equilibrium_P3}, the optimal $\mathbf{y}^{*}$ can be obtained as  $\mathbf{y}^{*} = \arg \max_{\mathbf{y}\in \Delta} f(\mathbf{X}^{*}, \mathbf{y})$. Since $f(\mathbf{X}^{*}, \mathbf{y})=\sum_{m=1}^{M}y_{m}f_{m}(\mathbf{X}^{*})$ and $\sum_{m=1}^{M}y_{m} = 1$, $\max_{\mathbf{y}\in \Delta} f(\mathbf{X}^{*}, \mathbf{y})$ is equivalent to $\max_{m} f_{m}(\mathbf{X}^{*})$ and there exists an optimal solution $\mathbf{y}^{*}$ being an unit vector.  

    As for optimizing indicator matrix $\mathbf{X}$, it is obvious from the second inequality of \eqref{Nash_equilibrium_P3} that $\mathbf{X}^{*} = \arg \min_{\mathbf{X}\in \varOmega} f(\mathbf{X}, \mathbf{y}^{*})$. Note that it has been proven in \cite{tight} that $\min_{\mathbf{X}\in \varOmega} \sum_{m=1}^{M}f_{m}(\mathbf{X})$ and $\min_{\mathbf{X}\in\varOmega_{01}} \sum_{m=1}^{M}f_{m}(\mathbf{X})$ are equivalent in the sense that both of them have the same optimal value and share a common optimal solution as long as $f_{m}(\mathbf{X})$ is concave in $\mathbf{X}$. Since $\mathbf{y}^{*}$ is a fixed non-negative real vector and $f(\mathbf{X}, \mathbf{y}^{*}) = \sum_{m=1}^{M}y_{m}^{*}f_{m}(\mathbf{X})$, $\min_{\mathbf{X}\in \varOmega} f(\mathbf{X}, \mathbf{y}^{*})$ and $\min_{\mathbf{X}\in\varOmega_{01}} f(\mathbf{X}, \mathbf{y}^{*})$ are also equivalent due to Lemma \ref{fm_concave}.
   
    It can be then concluded that the optimal solution $\mathbf{X}^{*}$ to the relaxed problem $\mathcal{P}6$ is also the optimal solution to problem $\mathcal{P}7$, and both the objective functions of $\mathcal{P}7$ and $\mathcal{P}6$ have the same optimal value. The proof is completed.
\end{IEEEproof}

Based on the above analysis, it is straightforward to show that $\mathcal{P}3$ is equivalent to $\mathcal{P}6$ as long as $\mathcal{P}6$ has a Nash equilibrium. Therefore, to establish the equivalence between $\mathcal{P}2$ and $\mathcal{P}6$, it remains to prove the equivalence between $\mathcal{P}2$ and $\mathcal{P}3$, which can be established by the following Lemma \ref{P2_P3_equivalence}.
\begin{lemma}\label{P2_P3_equivalence}
    When $\alpha \rightarrow \infty$, $\mathcal{P}2$ and $\mathcal{P}3$ are equivalent in the sense that the optimal values of the objective functions are the same and there exists a common optimal solution.
\end{lemma}
\begin{IEEEproof}[$Proof$]
    The equivalence of $\mathcal{P}3$, $\mathcal{P}6$ and $\mathcal{P}7$ indicates that the optimal solution to $\mathcal{P}3$, denoted by $\mathbf{X}^{*}_{3}$, satisfies the constraint of $\mathcal{P}7$, i.e., $\mathbf{X} \in \varOmega_{01}$, which is also the constraint of $\mathcal{P}2$.
    
    Let $v_{2}$ and $v_{3}$ be the optimal values of $\mathcal{P}2$ and $\mathcal{P}3$, respectively, thus $v_{3} = F_{\alpha}(\mathbf{X}^{*}_{3})$ and $v_{2} \leq F_{\alpha}(\mathbf{X}^{*}_{3})$ because $\mathbf{X}^{*}_{3}$ is a feasible solution to $\mathcal{P}2$, resulting to $v_{2} \leq v_{3}$. On the other hand, the constraint set for $\mathcal{P}2$ is a subset of that for $\mathcal{P}3$, so we have $v_{2} \geq v_{3}$. Finally, $\mathcal{P}2$ and $\mathcal{P}3$ have the same optimal value $v_{2} = v_{3}$, and exist a common optimal solution $\mathbf{X}^{*}_{3}$. The proof is completed.
\end{IEEEproof}
Therefore, by using $\mathcal{P}3$ and $\mathcal{P}7$ as bridges, we establish the equivalence between $\mathcal{P}2$ and $\mathcal{P}6$.

\section{Proof of Lemma  \ref{find_y_t}}\label{find_y_t_proof}
\begin{IEEEproof}
    By applying the KKT conditions \cite{kuhn2013nonlinear} and after some calculations, the problem is transformed to finding an integer $N$ such that
    \begin{equation}
        \label{find_N_transformed}
        \begin{aligned}
        a_{N} - \dfrac{1}{N}\left(\sum_{n=1}^{N}a_{n} -1\right) \geq 0, \\
        a_{N+1} - \dfrac{1}{N}\left(\sum_{n=1}^{N}a_{n} -1\right) < 0,
        \end{aligned}
    \end{equation}
    which can be found by iterating $N$ from 1 to $M$. 
    
    When $N = 1$, the first equation of \eqref{find_N_transformed} is always satisfied. If the second equation of \eqref{find_N_transformed} is also satisfied, $N=1$ is a solution. Otherwise, we will check whether the next iteration meets \eqref{find_N_transformed}. Assume that there exists an integer $\widetilde{N}$ such that
    \begin{equation}
        \label{find_N_transformed_first_term}
        a_{\widetilde{N}} - \dfrac{1}{\widetilde{N}}\left(\sum_{n=1}^{\widetilde{N}}a_{n} -1\right) \geq 0,
    \end{equation}
    and
    \begin{equation}
        a_{\widetilde{N}+1} - \dfrac{1}{\widetilde{N}}\left(\sum_{n=1}^{\widetilde{N}}a_{n} -1\right) \geq 0.
    \end{equation}
    In the next iteration, we can obtain from  \eqref{find_N_transformed_first_term} that 
    \begin{equation*}
        \begin{aligned}
        &\ a_{\widetilde{N}+1} - \dfrac{1}{\widetilde{N}+1}\left(\sum_{n=1}^{\widetilde{N}+1}a_{n} -1\right)\\
        =&\ \dfrac{\widetilde{N}}{\widetilde{N}+1}\left[a_{\widetilde{N}+1} - \dfrac{1}{\widetilde{N}}\left(\sum_{n=1}^{\widetilde{N}}a_{n} -1\right)\right]\\
        \geq&\ 0.
        \end{aligned}
    \end{equation*}
    Therefore, the iteration can proceed further, i.e., determining whether \eqref{find_N_transformed} is valid when $N = \widetilde{N} + 1$.
    This process terminates in at most $M$ steps. If $a_{M} - \left(\sum_{n=1}^{M}a_{n} -1\right)/M \geq 0$, we simply have $N = M$. 
    
    It can be seen that in all cases, we can always find a suitable $N$ that satisfies \eqref{find_N_transformed}. Thus, the optimal $\mathbf{y}^{*}$ can be obtained from \eqref{find_N} according to $N$. The proof is completed.
\end{IEEEproof}

\section{Proof of Lemma  \ref{Hadamard_Cauchy}}\label{Hadamard_Cauchy_proof}
\begin{IEEEproof}
    First, let us prove that $\xi = \sqrt{m}$ satisfies \eqref{Hadamard_Cauchy_xi}. Let $\mathbf{v} = \sum_{j=1}^{m}v_{j}\mathbf{e}_{j}$ be any vector with $\mathbf{e}_{1}, \cdots, \mathbf{e}_{m}$ being the standard basis vectors and $a = \max_{i,j}|\mathbf{A}(i,j)|$. Note that for each $j$, we have $\|(\mathbf{A} \circ \mathbf{B})\mathbf{e}_{j}\|_{2} \leq a\|\mathbf{B}\|_{2}$, since $\|\mathbf{B}\mathbf{e}_{j}\|_{2} \leq \|\mathbf{B}\|_{2}$ and $(\mathbf{A} \circ \mathbf{B})\mathbf{e}_{j}$ can be obtained by multiplying each entry of $\mathbf{B}\mathbf{e}_{j}$ by a scalar of size at most $a$. Thus
    \begin{equation}
        \|(\mathbf{A} \circ \mathbf{B})\mathbf{v}\|_{2} \!\leq \!\sum_{j=1}^{m}|v_{j}|\|(\mathbf{A} \circ \mathbf{B})\mathbf{e}_{j}\|_{2} \leq a\|\mathbf{B}\|_{2}\sum_{j=1}^{m}|v_{j}|.
    \end{equation}
    By Cauchy-Schwarz, we have $\sum_{j=1}^{m}|v_{j}| \leq \sqrt{m}\|\mathbf{v}\|_2$, and hence $\|\mathbf{A} \circ \mathbf{B}\|_{2} \leq a\sqrt{m}\|\mathbf{B}\|_{2}$ and $\xi = \sqrt{m}$ satisfies \eqref{Hadamard_Cauchy_xi}.

    To prove $\xi = \sqrt{m}$ is optimal, let $\varOmega$ be a primitive $n$-th root of unity and let $\mathbf{B}(i,j) = \varOmega^{(i-1)(j-1)}/\sqrt{n}, i = 1, \cdots, n, j= 1, \cdots, m$. The columns of $\mathbf{B}$ are orthogonal and $\|\mathbf{B}\|_{2} = 1$. Now let $\mathbf{A}(i, j) = \varOmega^{-(i-1)(j-1)}$, so the entries of $\mathbf{A} \circ \mathbf{B}$ are all $1/\sqrt{n}$, which indicates $\|\mathbf{A} \circ \mathbf{B}\|_{2} = \sqrt{m}$. As $|\mathbf{A}(i,j)| = 1,\ \forall i,j$, we have $\xi \geq \sqrt{m}$, thus the minimum value of $\xi$ is $\sqrt{m}$. The proof is completed.
\end{IEEEproof}

\bibliographystyle{IEEEtran.bst}
\bibliography{ref/refs.bib}

@book{2,
  title={Fundamentals of Wireless Communication},
  author={David Tse and Pramod Viswanath},
  year={2005},
  publisher={United States of America by Cambridge University Press, New York}
}

@ARTICLE{8007415,
  author={L. {Dai} and B. {Bai}},
  journal={IEEE Trans. Wireless Commun.},
  title={Optimal Decomposition for Large-Scale Infrastructure-Based Wireless Networks}, 
  month={Aug.},
  year={2017},
  volume={16},
  number={8},
  pages={4956-4969}}

@ARTICLE{7827017,
  author={H. Q. {Ngo} and A. {Ashikhmin} and H. {Yang} and E. G. {Larsson} and T. L. {Marzetta}},
  journal={IEEE Trans. Wireless Commun.},
  title={Cell-Free Massive {MIMO} Versus Small Cells}, 
  month={Jan.},
  year={2017},
  volume={16},
  number={3},
  pages={1834-1850}}

@inproceedings{9838756,
  author={{J. Wang, L. Dai, L. Yang, and B. Bai}},
  booktitle={Proc. IEEE ICC}, 
  title={Rate-Constrained Network Decomposition for Clustered Cell-Free Networking}, 
  month={May},
  year={2022},
  volume={},
  number={},
  pages={2549-2554},
  doi={10.1109/ICC45855.2022.9838756}}

@ARTICLE{10014663,
  author={Wang, Junyuan and Dai, Lin and Yang, Lu and Bai, Bo},
  journal={IEEE Trans. Wireless Commun.}, 
  title={Clustered Cell-Free Networking: A Graph Partitioning Approach}, 
  month={Aug.},
  year={2023},
  volume={22},
  number={8},
  pages={5349-5364},
  doi={10.1109/TWC.2022.3233444}}

@ARTICLE{9754261,
  author={Yang, Lu and Li, Ping and Dong, Miaomiao and Bai, Bo and Zaporozhets, Dmitry and Chen, Xiang and Han, Wei and Li, Baochun},
  journal={IEEE Trans. Wireless Commun.}, 
  title={{$C^2$}: A Capacity-Centric Architecture Toward Future Wireless Networking}, 
  month={Oct.},
  year={2022},
  volume={21},
  number={10},
  pages={8134-8147},
  doi={10.1109/TWC.2022.3164286}}

@inproceedings{9839089,
  author={Deng, Chaowen and Yang, Lu and Wu, Hao and Zaporozhets, Dmitry and Dong, Miaomiao and Bai, Bo},
  booktitle={Proc. IEEE ICC}, 
  title={{CGN}: A Capacity-Guaranteed Network Architecture for Future Ultra-Dense Wireless Systems}, 
  month={May},
  year={2022},
  volume={},
  number={},
  pages={1853-1858},
  doi={10.1109/ICC45855.2022.9839089}}

@inproceedings{nouiehed2019solving,
  title={Solving a class of non-convex min-max games using iterative first order methods},
  author={Nouiehed, Maher and Sanjabi, Maziar and Huang, Tianjian and Lee, Jason D and Razaviyayn, Meisam},
  booktitle={Proc. NeurIPS},
  volume={32},
  month={Dec.},
  year={2019}
}

@article{demir2021foundations,
  title={Foundations of user-centric cell-free massive {MIMO}},
  author={Demir, {\"O}zlem Tugfe and Bj{\"o}rnson, Emil and Sanguinetti, Luca and others},
  journal={Found. Trends Signal Process.},
  volume={14},
  number={3-4},
  pages={162--472},
  year={2021},
  publisher={Now Publishers, Inc.}
}

@inproceedings{tight,
  title={Tight Differentiable Relaxation of Sum Ergodic Capacity Maximization for Clustered Cell-Free Networking},
  author={Ren, Boxiang and Hao, Han and Lyu, Ziyuan and Peng, Jingchen and Wang, Junyuan and Wu, Hao},
  booktitle={Proc. IEEE ISIT},
  pages={2448--2453},
  year={2024},
  month={Jul.}
}

@article{pan2021efficient,
  title={An efficient algorithm for nonconvex-linear minimax optimization problem and its application in solving weighted maximin dispersion problem},
  author={Pan, Weiwei and Shen, Jingjing and Xu, Zi},
  journal={Comput. Optim. Appl.},
  volume={78},
  pages={287--306},
  year={2021},
  publisher={Springer}
}

@article{Danskin,
 ISSN = {00361399},
 author = {John M. Danskin},
 journal = {SIAM J. Appl. Math},
 number = {4},
 pages = {641--664},
 publisher = {Society for Industrial and Applied Mathematics},
 title = {The Theory of Max-Min, with Applications},
 urldate = {2024-06-11},
 volume = {14},
 year = {1966}
}

@inproceedings{boyle1986method,
  title={A method for finding projections onto the intersection of convex sets in {H}ilbert spaces},
  author={Boyle, James P and Dykstra, Richard L},
  booktitle={Advances in Order Restricted Statistical Inference},
  pages={28--47},
  year={1986},
  organization={Springer}
}

@incollection{kuhn2013nonlinear,
  title={Nonlinear programming},
  author={Kuhn, Harold W and Tucker, Albert W},
  booktitle={Traces and emergence of nonlinear programming},
  pages={247--258},
  year={2013},
  publisher={Springer}
}

@article{bauschke2000dykstras,
  title={Dykstra's algorithm with {B}regman projections: A convergence proof},
  author={Bauschke, Heinz H and Lewis, Adrian S},
  journal={Optimization},
  volume={48},
  number={4},
  pages={409--427},
  year={2000},
  publisher={Taylor \& Francis}
}

@inproceedings{malinen2014balanced,
  title={Balanced k-means for clustering},
  author={Malinen, Mikko I and Fr{\"a}nti, Pasi},
  booktitle={Structural, Syntactic, and Statistical Pattern Recognition: Joint IAPR International Workshop, S+ SSPR 2014, Joensuu, Finland, August 20-22, 2014. Proceedings},
  pages={32--41},
  year={2014},
  organization={Springer}
}

@article{hieu2022two,
  title={Two {B}regman projection methods for solving variational inequalities},
  author={Hieu, Dang Van and Reich, Simeon},
  journal={Optimization},
  volume={71},
  number={7},
  pages={1777--1802},
  year={2022},
  publisher={Taylor \& Francis}
}

@article{kelly1998rate,
  title={Rate control for communication networks: shadow prices, proportional fairness and stability},
  author={Kelly, Frank P and Maulloo, Aman K and Tan, David Kim Hong},
  journal={Journal of the Operational Research society},
  volume={49},
  number={3},
  pages={237--252},
  month={Mar.},
  year={1998},
  publisher={Taylor \& Francis}
}

@ARTICLE{879343,
  author={Mo, J. and Walrand, J.},
  journal={IEEE/ACM Trans. Networking}, 
  title={Fair end-to-end window-based congestion control}, 
  year={2000},
  month={Oct.},
  volume={8},
  number={5},
  pages={556-567},
  doi={10.1109/90.879343}}

@article{Hachem_2007,
   title={Deterministic equivalents for certain functionals of large random matrices},
   volume={17},
   ISSN={1050-5164},
   number={3},
   journal={Ann. Appl. Probab.},
   publisher={Institute of Mathematical Statistics},
   author={Hachem, Walid and Loubaton, Philippe and Najim, Jamal},
   year={2007},
   month={June},
   pages={875–930},
}

@article{liu2014graph,
  title={Graph matching by simplified convex-concave relaxation procedure},
  author={Liu, ZhiYong and Qiao, Hong and Yang, Xu and Hoi, Steven CH},
  journal={Int. J. Comput. Vision},
  volume={109},
  pages={169--186},
  year={2014},
  publisher={Springer}
}

@article{liu2013gnccp,
  title={{GNCCP}—graduated nonconvexityand concavity procedure},
  author={Liu, ZhiYong and Qiao, Hong},
  journal={IEEE Trans. Pattern Anal. Mach. Intell.},
  volume={36},
  number={6},
  pages={1258--1267},
  year={2013},
  publisher={IEEE}
}

@incollection{martello1987linear,
  title={Linear assignment problems},
  author={Martello, Silvano and Toth, Paolo},
  booktitle={North-Holland Mathematics Studies},
  volume={132},
  pages={259--282},
  year={1987},
  publisher={Elsevier}
}

@article{frank1956algorithm,
  title={An algorithm for quadratic programming},
  author={Frank, Marguerite and Wolfe, Philip and others},
  journal={Naval research logistics quarterly},
  volume={3},
  number={1-2},
  pages={95--110},
  year={1956},
  publisher={Wiley Subscription Services, Inc., A Wiley Company New York}
}

@article{fan2024optimal,
author = {Jie, Fan and Wu, Tianhao and Wu, Hao},
title = {The Optimal Production Transport: Model and Algorithm},
journal = {East Asian Journal on Applied Mathematics},
year = {2025},
volume = {15},
number = {4},
pages = {787--816\color{black}},
issn = {2079-7370},
doi = {https://doi.org/10.4208/eajam.2024-095.140824}
}

@article{pokutta2024frank,
  title={The Frank-Wolfe algorithm: a short introduction},
  author={Pokutta, Sebastian},
  journal={Jahresbericht der Deutschen Mathematiker-Vereinigung},
  volume={126},
  number={1},
  pages={3--35},
  year={2024},
  publisher={Springer}
}

@inproceedings{ren2024sequential,
  title={A Sequential Min {$K$}-Cut Approach for Sum Rate Maximization of Clustered Cell-Free Networking},
  author={Ren, Boxiang and Deng, Chaowen and Hao, Han and Wu, Hao and Wang, Junyuan},
  booktitle={Proc. IEEE ICC},
  month={Jun.},
  pages={4900--4905},
  year={2024},
  organization={}
}

@ARTICLE{9186144,
  author={Razaviyayn, Meisam and Huang, Tianjian and Lu, Songtao and Nouiehed, Maher and Sanjabi, Maziar and Hong, Mingyi},
  journal={IEEE Signal Process Mag.}, 
  title={Nonconvex Min-Max Optimization: Applications, Challenges, and Recent Theoretical Advances}, 
  year={2020},
  month={Sep.},
  volume={37},
  number={5},
  pages={55-66},
}

@inproceedings{zeng2024tunable,
  title={Tunable Weighted Kernel $k$-Means for Clustered Cell-Free Networking Acceleration and Beam On-Off Control},
  author={Zeng, Xiankun and Wang, Junyuan and Yue, Ke and Dong, Miaomiao and Bai, Bo},
  booktitle={Proc. IEEE ICC},
  pages={4311--4316},
  year={2024},
  month={June}
}

@ARTICLE{10766356,
  author={Xia, Funing and Wang, Junyuan and Dai, Lin},
  journal={IEEE Trans. Wireless Commun.}, 
  title={Optimizing Clustered Cell-Free Networking for Sum Ergodic Capacity Maximization With Joint Processing Constraint}, 
  year={2025},
  month={Jan.},
  volume={24},
  number={1},
  pages={571-584},
  doi={10.1109/TWC.2024.3496733}}

@ARTICLE{7069272,
  author={Timotheou, Stelios and Krikidis, Ioannis},
  journal={IEEE Signal Process Lett.}, 
  title={Fairness for Non-Orthogonal Multiple Access in 5{G} Systems}, 
  year={2015},
  month={Oct.},
  volume={22},
  number={10},
  pages={1647-1651},
  doi={10.1109/LSP.2015.2417119}}

@ARTICLE{7460209,
  author={Liu, Yuanwei and Elkashlan, Maged and Ding, Zhiguo and Karagiannidis, George K.},
  journal={IEEE Commun. Lett.}, 
  title={Fairness of User Clustering in {MIMO} Non-Orthogonal Multiple Access Systems}, 
  year={2016},
  month={Jul.},
  volume={20},
  number={7},
  pages={1465-1468},
  doi={10.1109/LCOMM.2016.2559459}}

@ARTICLE{7986959,
  author={Xu, Peng and Cumanan, Kanapathippillai},
  journal={IEEE J. Sel. Areas Commun.}, 
  title={Optimal Power Allocation Scheme for Non-Orthogonal Multiple Access With $\alpha $ -Fairness}, 
  year={2017},
  month={Oct.},
  volume={35},
  number={10},
  pages={2357-2369},
  doi={10.1109/JSAC.2017.2729780}}

@article{cuturi2013sinkhorn,
  title={Sinkhorn distances: Lightspeed computation of optimal transport},
  author={Cuturi, Marco},
  journal={Advances in neural information processing systems},
  volume={26},
  year={2013}
}

@ARTICLE{7506136,
  author={Xu, Peng and Yuan, Yi and Ding, Zhiguo and Dai, Xuchu and Schober, Robert},
  journal={IEEE Trans. Wireless Commun.}, 
  title={On the Outage Performance of Non-Orthogonal Multiple Access With 1-bit Feedback}, 
  year={2016},
  month={Oct.},
  volume={15},
  number={10},
  pages={6716-6730},
  doi={10.1109/TWC.2016.2587880}}

@ARTICLE{6517050,
  author={SHI, Huaizhou and Prasad, R. Venkatesha and Onur, Ertan and Niemegeers, I.G.M.M.},
  journal={IEEE Commun. Surv. Tutorials}, 
  title={Fairness in Wireless Networks:Issues, Measures and Challenges}, 
  year={2014},
  month={First Quarter},
  volume={16},
  number={1},
  pages={5-24},
  doi={10.1109/SURV.2013.050113.00015}}

@ARTICLE{10706102,
  author={Zhu, Guangyu and Mu, Xidong and Guo, Li and Huang, Ao and Xu, Shibiao},
  journal={IEEE Internet Things J.}, 
  title={Enhancing User Fairness in Wireless Powered Communication Networks With {STAR-RIS}}, 
  year={2025},
  month={Feb.},
  volume={12},
  number={3},
  pages={2659-2673},
  doi={10.1109/JIOT.2024.3475215}}

@ARTICLE{10505156,
  author={Göttsch, Fabian and Osawa, Noboru and Kanno, Issei and Ohseki, Takeo and Caire, Giuseppe},
  journal={IEEE Trans. Wireless Commun.}, 
  title={Fairness Scheduling in User-Centric Cell-Free Massive {MIMO} Wireless Networks}, 
  year={2024},
  month={Sep.},
  volume={23},
  number={9},
  pages={11942-11957},
  doi={10.1109/TWC.2024.3386802}}

@ARTICLE{10714036,
  author={Sun, Gang and Wang, Yuhui and Yu, Hongfang and Guizani, Mohsen},
  journal={IEEE Trans. Serv. Comput.}, 
  title={Proportional Fairness-Aware Task Scheduling in Space-Air-Ground Integrated Networks}, 
  year={2024},
  month={Nov.-Dec.},
  volume={17},
  number={6},
  pages={4125-4137},
  doi={10.1109/TSC.2024.3478730}}

@book{boyd2004convex,
  title={Convex optimization},
  author={Boyd, Stephen P and Vandenberghe, Lieven},
  year={2004},
  publisher={Cambridge university press}
}

@article{jain1984quantitative,
  title={A quantitative measure of fairness and discrimination},
  author={Jain, Rajendra K and Chiu, Dah-Ming W and Hawe, William R and others},
  journal={Eastern Research Laboratory, Digital Equipment Corporation, Hudson, MA},
  volume={21},
  number={1},
  pages={2022--2023},
  year={1984}
}

@ARTICLE{8869705,
  author={Saad, Walid and Bennis, Mehdi and Chen, Mingzhe},
  journal={IEEE Network}, 
  title={A Vision of {6G} Wireless Systems: Applications, Trends, Technologies, and Open Research Problems}, 
  year={2020},
  month={May/Jun.},
  volume={34},
  number={3},
  pages={134-142},
  doi={10.1109/MNET.001.1900287}
}

@ARTICLE{5770660,
  author={You, Xiaohu and Wang, Dongming and Zhu, Pengcheng and Sheng, Bin},
  journal={IEEE J. Sel. Areas Commun.}, 
  title={Cell Edge Performance of Cellular Mobile Systems}, 
  year={2011},
  month={Jun.},
  volume={29},
  number={6},
  pages={1139-1150},
  doi={10.1109/JSAC.2011.110603}
}

@ARTICLE{340450,
  author={Wyner, A.D.},
  journal={IEEE Trans. Inf. Theory}, 
  title={Shannon-theoretic approach to a {Gaussian} cellular multiple-access channel}, 
  year={1994},
  month={Nov.},
  volume={40},
  number={6},
  pages={1713-1727},
  doi={10.1109/18.340450}
}

@ARTICLE{5706317,
  author={Irmer, Ralf and Droste, Heinz and Marsch, Patrick and Grieger, Michael and Fettweis, Gerhard and Brueck, Stefan and Mayer, Hans-Peter and Thiele, Lars and Jungnickel, Volker},
  journal={IEEE Commun. Mag.}, 
  title={Coordinated multipoint: Concepts, performance, and field trial results}, 
  year={2011},
  month={Feb.},
  volume={49},
  number={2},
  pages={102-111},
  doi={10.1109/MCOM.2011.5706317}
}

@ARTICLE{9064545,
  author={Björnson, Emil and Sanguinetti, Luca},
  journal={IEEE Trans. Commun.}, 
  title={Scalable Cell-Free Massive {MIMO} Systems}, 
  year={2020},
  month={Jul.},
  volume={68},
  number={7},
  pages={4247-4261},
  doi={10.1109/TCOMM.2020.2987311}
}

@ARTICLE{8845768,
  author={Björnson, Emil and Sanguinetti, Luca},
  journal={IEEE Trans. Wireless Commun.}, 
  title={Making Cell-Free Massive {MIMO} Competitive With {MMSE} Processing and Centralized Implementation}, 
  year={2020},
  month={Jan.},
  volume={19},
  number={1},
  pages={77-90},
  doi={10.1109/TWC.2019.2941478}
}

@article{interdonato2019ubiquitous,
  title={Ubiquitous cell-free massive {MIMO} communications},
  author={Interdonato, Giovanni and Bj{\"o}rnson, Emil and Quoc Ngo, Hien and Frenger, P{\aa}l and Larsson, Erik G},
  journal={EURASIP J. Wireless Commun. Networking	},
  volume={2019},
  number={1},
  pages={1--13},
  year={2019},
  month={Aug.},
  publisher={Springer}
}

@ARTICLE{11220251,
  author={Mobini, Zahra and Gokceoglu, Ahmet Hasim and Wang, Li and Peters, Gunnar and Shin, Hyundong and Ngo, Hien Quoc},
  journal={IEEE Trans. Wireless Commun.}, 
  title={Cluster-Wise Processing in Fronthaul-Aware Cell-Free Massive {MIMO} Systems}, 
  year={2025},
  month={Oct.},
  volume={},
  number={},
  pages={1-1},
  doi={10.1109/TWC.2025.3624199}
}

@INPROCEEDINGS{5461911,
  author={Lan, Tian and Kao, David and Chiang, Mung and Sabharwal, Ashutosh},
  booktitle={Proc. IEEE INFOCOM}, 
  title={An Axiomatic Theory of Fairness in Network Resource Allocation}, 
  year={2010},
  month={Mar.},
  volume={},
  number={},
  pages={1-9},
  doi={10.1109/INFCOM.2010.5461911}}

@INPROCEEDINGS{10978229,
  author={Deng, Chaowen and Fan, Jie and Ren, Boxiang and Lyu, Ziyuan and Wang, Junyuan and Wu, Hao},
  booktitle={Proc. IEEE WCNC}, 
  title={Towards Load-Balanced Clustered Cell-Free Networking: A Tight Relaxation Approach}, 
  year={2025},
  month={Mar.},
  volume={},
  number={},
  pages={1-6},
  doi={10.1109/WCNC61545.2025.10978229}}

\end{document}